\documentclass{article}
\usepackage[preprint]{neurips_2020}

\usepackage[utf8]{inputenc} % allow utf-8 input
\usepackage[T1]{fontenc}    % use 8-bit T1 fonts
\usepackage{hyperref}       % hyperlinks
\usepackage{url}            % simple URL typesetting
\usepackage{booktabs}       % professional-quality tables
\usepackage{amsfonts}       % blackboard math symbols
\usepackage{nicefrac}       % compact symbols for 1/2, etc.
\usepackage{microtype}      % microtypography
\usepackage{multirow}%
\usepackage{subfigure}
\usepackage{tikz}
\usepackage{amsmath,amsthm,amsfonts,amssymb}%
\usepackage{graphicx}
\usetikzlibrary{positioning}
\usepackage{enumerate}
\usepackage{dsfont}
\usepackage{array}
\usepackage{bm}
\usepackage{color}
\usepackage{xcolor}
\definecolor{linkcolor}{RGB}{83,83,182}
\definecolor{citecolor}{RGB}{128,0,128}
\hypersetup{
    colorlinks=true,
    citecolor=citecolor,
    linkcolor=linkcolor,
    urlcolor=linkcolor}
\usepackage{changes}
\usepackage{paralist}

\usepackage[titletoc,title]{appendix}

\usepackage{titletoc}
\usepackage{calc}

\newtheorem{theorem}{Theorem}
\newtheorem{fact}{Fact}

\newtheorem{lemma}[theorem]{Lemma}

\newtheorem{definition}{Definition}

\theoremstyle{remark}
\newtheorem{remark}{Remark}

\newcount\LongProof
\newcommand{\longproof}[2]{
  \ifnum\LongProof=0{#1}\fi
  \ifnum\LongProof=1{#2}\fi
}

\newcommand\R{\mathbb{R}}

\newcommand\ind{\mathds{1}}     %indicators

\DeclareMathOperator{\MSE}{MSE}
\newcommand{\point}{\,\cdot\,}

\DeclareMathOperator{\Cov}{Cov}

\newcommand{\diff}{\mathrm{d}}

\newcommand{\bu}{\bm{u}}

\newcommand{\bY}{\bm{Y}}
\newcommand{\by}{\bm{y}}

\allowdisplaybreaks[1]

\newcommand{\indep}{\perp\hspace{-.25cm}\perp}

\definecolor{auburn}{rgb}{0.43, 0.21, 0.1}
\definecolor{britishracinggreen}{rgb}{0.0, 0.26, 0.15}
\definecolor{burntumber}{rgb}{0.54, 0.2, 0.14}
\definecolor{carmine}{rgb}{0.59, 0.0, 0.09}

\title{Conditional independence testing via weighted partial copulas and nearest neighbors}

\author{%
Pascal Bianchi \qquad Kevin Elgui \qquad François Portier\thanks{Corresponding address: \texttt{francois.portier@gmail.com}}
    \\ \ \\
  Télécom Paris\\
  Institut Polytechnique de Paris
}

\begin{document}

\maketitle

\begin{abstract}
This paper introduces the \textit{weighted partial copula} function for testing conditional independence. The proposed test procedure results from these two ingredients: (i) the test statistic is an explicit Cramer-von Mises transformation of the \textit{weighted partial copula}, (ii) the regions of rejection are computed using a bootstrap procedure which mimics conditional independence by generating samples from the product measure of the estimated conditional marginals. Under mild conditions on the estimated conditional marginals, rates of convergence for the \textit{weighted partial copula process} and the \textit{test statistic} are established. These high-level conditions are shown to be valid for the popular \textit{nearest neighbors} estimate. Finally, an experimental section demonstrates that the proposed test has competitive power compared to recent state-of-the-art methods such as kernel-based test. 
\end{abstract}

\section{Introduction}

Let $(Y_1, Y_2,X)$ be a triple of real random variables. We say that $Y_1$ and $Y_2$ are conditionally independent given $X$ if $\forall  (y_1,y_2, x) \in \mathbb R^3$:
\begin{equation}
    \Pr(Y_1\leq y_1,\,Y_2\leq y_2\mid X= x ) 
    = \Pr(Y_1 \leq y_1\mid X= x )  \Pr( Y_2\leq y_2\mid X= x ). 
\end{equation}
%\begin{align*}
%\Pr(Y_1\leq y_1,\,Y_2\leq y_2\mid X= x ) = \Pr(Y_1 \leq y_1\mid X= x ) \Pr( Y_2\leq y_2\mid X= x ),\qquad \forall (y_1,y_2, x)\in \mathbb R^3.
%\end{align*}
This is denoted by $Y_1 \indep Y_2~|~X$ and roughly speaking, it means that for a given value of $X$, the knowledge of $Y_1$ does not provide any further information on $Y_2$ (and vice versa). Determining conditional independence has become in the recent years a fundamental question in statistics and machine learning. For instance, it plays a key role in defining 
%computational statistics since it is an essential component for building 
\textit{graphical models} \citep{koller2009probabilistic,bach+j:2003};  see also \cite{markowetz2007inferring} for a study specific to cellular networks. Moreover the concept of conditional independence lies at the core of \textit{sufficient dimension reduction} methods \citep{li:2018} and is useful to conduct variable selection in regression \citep{lee+l+z:2016}. Finally, conditional independence is relevant in many application fields such as economy \citep{huber2015test} or psychometry \citep{bell1988conditional}. This paper proposes new statistical tests to assess conditional independence.

The approach taken is related to the well-studied problem of (unconditional) independence testing, in which the most intuitive way to proceed is perhaps to compute a distance between the estimated joint distribution and the product of the estimated marginals \citep{hoeffding:1948}. Inspired by \cite{kendall:1948}, rank-based statistics have been extensively used in independence testing \citep{ruymgaart:1974,ruschendorf:1976,ruymgaart:1978}. Because rank-based statistics do not depend on the marginals, they have appeared as a key tool for modelling the joint distribution of random variables without being affected by their margins. This has led to the introduction of the \textit{copula function} \citep{deheuvels:1981}, defined as the cumulative distribution function associated to the ranks. We refer to \cite{fermanian+r+w:2004,segers:2012} for recent studies on the estimation of the copula function. The copula function, which in principle measures the dependency between random variables has been used with success in independence testing \citep{genest+r:2004,genest:2006}. Because the asymptotic distribution of the copula function is difficult to estimate, the related bootstrap estimate properties are of prime interest for inference \citep{fermanian+r+w:2004,remillard:2009,bucher+d:2010}.

%\pb{Paragraphe qui suit plus difficile a lire. Peut on preciser ?} Compared to independence, conditional independence is a property much more difficult to detect. This is because conditional estimation \pb{de quoi ?} is attached \pb{peu precis } to slower convergence rates than the classical $n^{-1/2}$ convergence rate obtained when estimating the (unconditional) distribution function. Discussing convergence rates is relevant to provide insights into the types of local alternative against which the test will be robust \pb{preciser}. In particular, 
%using the $n^{-1/2}$-convergence rate of the copula estimate,  \cite{genest:2006} shows that their copula-based test is robust against $n^{-1/2}$-departures from the independence setting. In the case of conditional independence, it is not clear if such  local power properties are achievable due to the deterioration of the convergence rate. The previous observation relates to a large literature dealing with \textit{conditional moment restrictions}; see \cite{lavergne:2013} and the reference therein; in which integral based tests are shown to be robust to $n^{-1/2}$-departure from the null hypothesis. We refer also to \cite{stute:1997} for some tests on the regression function and a discussion on this issue.

The conditional copula of $Y_1$ and $Y_2$ given $X$ is defined in the same way as the copula of  $Y_1$ and $Y_2$ but uses the conditional distribution of $Y_1$ and $Y_2$ given $X$ in place of the joint distribution of  $Y_1$ and $Y_2$. Compared to the copula, the conditional copula captures the conditional dependency between random variables and is thus useful to build conditional dependency measures \citep{gijbels+v+o:2011:csda,fukumizu2007kernel,derumigny2020conditional}. Therefore, as in the case of independence testing, the conditional copula appears to be a relevant tool for building statistical test of conditional independence. This has been pointed out as a an ``interesting open issue'' in \citep[Section 4]{veraverbeke+o+g:2011}.
%While the asymptotic properties of the conditional copula estimator are well known \citep{veraverbeke+o+g:2011}, only a few research works are dedicated to the use of copulas in conditional independence testing (see below).

In this work, a new statistical test procedure, called the weighted partial copula test is investigated to assess conditional independence.
%inspired by the \textit{conditional moment restriction} literature (previously discussed), 
The proposed approach follows from the use of an integrated criterion, the \textit{weighted partial copula}, a function that equals $0$ if and only if conditional independence holds. 
Given estimators of the conditional marginals of $Y_1$ and $Y_2$ given $X$, the \textit{empirical weighted partial copula} is introduced to estimate the weighted partial copula and the test statistic results from an easy-to-compute Cramer-von Mises transformation. The use of the weighted partial copula is motivated by the \textit{conditional moment restrictions} literature (see \cite{lavergne:2013} and the reference therein) in which integrated criteria, similar to the weighted partial copula, have been frequently used. Those criteria are interesting because even when they involve local estimates converging at a slower rate than $1 / \sqrt n$, their convergence rates are in many cases in $1 / \sqrt n$.
%attached to arate of convergence $n^{-1/2}$, which is the same as the one derived in the (unconditional) independence test, is notable because conditional copula estimates are known \citep{veraverbeke+o+g:2011} to converge at a slower rate, $(nh^d)^{-1/2}$ where $h$ is a smoothing parameter going to $0$

%From a theoretical standpoint, the use of an ``integrated'' criterion enables to establish, in a general nonparametric framework, a convergence rate of order $n^{-1/2}$ for the empirical weighted partial copula. More precisely, by using a smoothed local linear estimator for the conditional marginals, we obtain the weak convergence of the empirical weighted partial copula rescaled by $n^{1/2}$. The rate of convergence $n^{-1/2}$, which is the same as the one derived in the (unconditional) independence test, is notable because conditional copula estimates are known \citep{veraverbeke+o+g:2011} to converge at a slower rate, $(nh^d)^{-1/2}$ where $h$ is a smoothing parameter going to $0$. Note finally that integrated criterion for testing has been frequently used in the \textit{conditional moment restrictions} literature (see \cite{lavergne:2013} and the reference therein).

Inspired by the independence testing literature \citep{beran:2007,kojadinovic:2009}, the computation of the quantiles is made using a bootstrap procedure which generates bootstrap samples from the product of the marginal estimators to mimic the null hypothesis. Thanks to this bootstrap procedure, one is allowed to perform the weighted partial copula test using any marginal estimates as soon as one can generate from them. 
%A numerical experiment is proposed to compare our test to the one described in \cite{zhang2012kernel}. We demonstrate in this experimental part that the power of our test is competitive. 

The theoretical results of the paper are as follows. Under the uniform convergence of the estimated marginals, convergence rates are established for the empirical weighted partial copula and for the resulting test statistic. These results are interesting because they allow, in principle, to use any reasonable margins estimates when using the weighted partial copula test. A particular method that we promote is the flexible nearest neighbors approach. For this type of margins estimates, the uniform convergence condition is shown to be valid in a general framework. More specifically, we establish that the uniform error associated to the $k$-nearest neighbors estimate of the conditional margins is bounded by $1/\sqrt k + (k/n ) ^{1/d} $.  This new result extends, in a non trivial way, previous works on nearest neighbors regression estimates \citep{biau2010rate,jiang2019non}.

\noindent \textbf{Related literature.} Nonparametric testing for conditional independence has received an increasing interest the past few years \citep{li2020nonparametric}.  Some of the existing approaches are based on comparing the (estimated) conditional distributions involved in the definition of conditional independence.
The distributions can be compared using their conditional characteristic functions as in \citet{su:2007}, their conditional densities as proposed in \cite{su2008nonparametric}, or their conditional copulas as studied in \cite{bouezmarni:2012}. Unfortunately, the estimation of these conditional quantities are subjected to the well-known curse of dimensionality, i.e., the convergence rates are badly affected by the dimension of the conditioning variable. As a consequence, the power of the previous tests rapidly deteriorates if the conditioning variable has a large dimension. Note also \citet{bergsma2010nonparametric} that uses  partial copulas to derive the test statistic. Unfortunately, partial copulas fail to capture the whole conditional distribution and lead to detect a null hypothesis much larger than conditional independence.

Other approaches rely on the characterization of conditional independence using cross-covariance operators defined on reproducing kernel Hilbert spaces \citep{fukumizu2004dimensionality}. Extending the Hilbert-Schmidt independence criterion proposed in \citet{gretton2008kernel}, \cite{zhang2012kernel} defines a kernel-based conditional independence test (KCI-test) by estimating the cross-covariance operator.
A surge of recent research \citep{doran2014permutation,runge2017conditional,sen2017model} has focused on testing conditional independence using permutation-based tests. The work
of \cite{candes2018panning} had led to many conditional independence tests depending on the availability of an approximation to the distribution of $Y_1|X$, such as the conditional permutation test proposed in \cite{berrett2019conditional}. In \cite{sen2017model}, the authors propose to train a classifier (e.g., XGBoost) to distinguish between two samples, one is the original sample, another one is a bootstrap sample generated in a way that reflects conditional independence. According to the accuracy of the trained classifier the test rejects, or not, conditional independence. This is further referred to as the classifier based conditional independence test (CCI-test).

\noindent \textbf{Outline.}
In Section \ref{sec:def}, we introduce the weighted partial copula test and provide implementation details including the mentioned bootstrap procedure.
In Section~\ref{sec:results}, we state the main results. In Section \ref{sec:numerical}, the theory is illustrated by numerical experiments. Our approach is compared to the ones described in \cite{zhang2012kernel} when facing simulated datasets. The proofs are given in a supplementary material file, as well an additional study dealing with functional connectivity.

%%% Local Variables: 
%%% mode: latex
%%% TeX-master: "../main"
%%% End: 

\section{The weighted partial copula test}\label{sec:def}

\subsection{Set-up and definitions}

Let $f_{X, \bY}$ be the density function (with respect to the Lebesgue measure) of the random triple $(X, \bY) = (X, Y_1, Y_2)\in \mathbb R^d \times \mathbb R\times \mathbb R$. Let $f_X$ and $S_X = \{ x \in \R : f_X(x) > 0 \}$ denote the density and the support of $X$, respectively. 
% Consider the componentwise ordering: $(Y_1,Y_2)\leq(y_1,y_2)=\by$ if and only if $Y_1\leq y_1$ and $Y_2\leq y_2 $. 
The conditional cumulative distribution function of $\bY$ given $X=x$ is given by $ \by \mapsto H(\by \mid x)$ for $x\in S_X$. The generalized inverse of a univariate distribution function $F$ is defined as 
$  F^-(u) = \inf \{ y \in \R : F(y) \ge u \}$, for all $u \in [0, 1]$, with the convention that $\inf \emptyset = + \infty$. Since $H(\point|x)$ is a continuous bivariate cumulative distribution function, its copula is given by the function
\begin{equation*}
  C( \bu \mid x) =
  H \left(    F_{1} ^- (u_1|x) , \;  F_{2}^- (u_2|x ) \mid  x  \right),
\end{equation*}
for $\bu = (u_1, u_2) \in [0, 1]^2$ and $x \in S_X$, where $F_1(\point|x)$ and $F_2(\point|x)$ are the margins of $H(\point|x)$. We are interested in testing the null hypothesis that $Y_1$ and $Y_2$ are conditionally independent given $X$, that is,
\begin{align*}
\mathcal H_0 \,  :\,   Y_1 \indep Y_2 | X .
\end{align*}
By definition \citep{dawid:1979}, $\mathcal H_0$ is equivalent to $H(\by | x ) = F_1(y_1|x) F_2(y_2|x)$, for every $  \by \in \R^2$ and almost every $x\in S_X$.
Using the conditional copula introduced before, it follows that
\begin{align*}
 \mathcal H_0  \ \Leftrightarrow \ C( \bu \mid x)  = u_1u_2 , \quad  \text{for every }  \bu \in  [0,1]^2,& \\
 \text{and almost every } x\in S_X.&
\end{align*}
Let $w: \mathbb R^d \to \mathbb R $ be a measurable function. The \textit{weighted partial copula} is given by, for every $\bu \in [0,1]^2$ and almost every $t\in  \mathbb R^d$,
\begin{align*}
W (\bu,t) = E \left[ (C( \bu \mid X)  - u_1u_2 )  w( t - X)  \right] .
\end{align*}
% \pb{Si $w(\cdot)$ n'est qu'intégrable (ce qui n'est pas precisé), $W$ n'est défini que pour tout $t$ presque partout,
% à moins que la densité de $X$ soit bornée (mais on ne fait pas d'hypothèse dessus). Soit on suppose $w$ OU $f$ bornée, soit on met des
% presque partout dans le lemme plus bas, vous ne pensez pas ?}
The proposed test follows from the observation, that $\mathcal H_0$ is satisfied if and only if the function $W$ is identically equal to $0$ under a certain (mild) condition on $w$. This is presented in the following lemma whose proof is given in the supplementary material.

\begin{lemma}\label{lemma:cara}
Suppose that $w: \mathbb R^d \to \mathbb R $ is integrable with respect to the Lebesgue measure and with a Fourier transform being non-zero almost everywhere, then $\mathcal H_0$ is equivalent to 
$  W (\bu,t)   = 0 $, for every $ \bu \in [0,1]^2 $ and almost every $t\in \mathbb R^d$.
% \pb{Idem, juste un peu gêné par l'intervalle fermé. Et aussi par les ``presque partout'' manquant peut etre ici, non ?}
% \fp{J'ai mis les presque partout aux endroits ou ils me semblaient manquants. Si tu es d'accord Pascal, tu peux enlever mes remarques et les tiennes}
\end{lemma}

\subsection{The test statistic}

In the following, we define a general estimator of $W$ relying on some empirical copula construction that works for any estimate of the marginals $F_1$ and $F_2$ (see Section \ref{sec:generic_example} for a typical example). That is, we first compute sample based observations of $F_{j}(Y_{j}|X)$, $j=1,2$, by estimating each marginal $F_j$. Those are usually called pseudo-observations. Second we define an estimate of $W$ based on the ranks of the pseudo-observation. For the sake of generality, the estimator used for the conditional marginals is left unspecified in the subsequent development.

Let $(X_i, Y_{i1}, Y_{i2})$, for $i \in \{1, \ldots, n\}$, be independent and identically distributed random vectors, with common distribution equal to the one of $(X, Y_1, Y_2)$. Estimate the conditional margins in some way, producing random functions $\hat{F}_{n,j}(\point|x)$, $j=1,2$, and then proceed with the pseudo-observations $\hat U_{ij} = \hat{F}_{n,j}(Y_{ij}|X_i)$. Let $\hat{G}_{nj}$, for $j \in \{1,2\}$, be the empirical cumulative distribution function of the pseudo-observations $(\hat U_{1j},\ldots, \hat U_{nj}) $, i.e. $
\hat G_{n,j}(u) = \frac{1}{n} \sum_{i=1}^n \ind_{\{ \hat U_{ij} \, \leq\, u\}},$ for $ u\in[0,1]. $ 
From a conditioning argument, the weighted partial copula is given by
\begin{align*}
W(\bu , t) = \mathbb E [ (\ind_{\{   F_{1}(Y_{1}|X)\, \leq  u_1 \}} \ind _{\{  F_{2}(Y_{2}|X)\, \leq  u_2  \}}  - 
u_1u_2 )  w(t- X) ].
\end{align*}
The previous expression suggests the introduction of following so-called the \textit{empirical weighted partial copula}, given by
\begin{align*}
\hat W _n(\bu , t) = \frac{1}{n} \sum_{i=1}^n  \Big( \ind _{\{ \hat U_{i1} \leq \hat G^{-}_{n,1}(u_1) \}} \ind _{\{ \hat U_{i2}\leq \hat G^{-}_{n,2}(u_2)\}}   -
u_1u_2 \Big) w(t - X_i).
\end{align*}
The use of the transform $G^{-}_{n,1}$ and $G^{-}_{n,2}$ implies that $\hat W _n$ depends on $\hat U_{i1}$ and $\hat U_{i2}$ only through their ranks (rank with respect to the natural ordering). Indeed, because  $\hat G_{n,j}$ is a càd-làg function with jumps $1/n$ at each $\hat U_{ij}$, it holds that $\hat U_{ij} \,  \leq \, \hat G^{-}_{n,j}(u_j) $ is equivalent to $(\hat R _ {ij}-1)/n < u_j$, where $\hat R_{ij} = n \hat G_{n,j}(\hat U_{ij} )$ is the rank of $\hat U_{ij}$ among the sample $(\hat U_{1j},\ldots, \hat U_{nj}) $. Hence, we have
\begin{align*}
\hat W _n  (\bu , t)= \frac{1}{n} \sum_{i=1}^n \Big( \ind _{\{ (\hat R _ {i1}-1) \,  < \, n u_1 \}}\ind _{\{ (\hat R_{i2} - 1) \, < \,  n u_2\}}    \nonumber  - u_1u_2\Big) w(t-X_i).
\end{align*}
The test statistic is given by
\begin{equation}
\hat T_n = \int_{[0,1]^2\times \mathbb R^d}  \hat W _n  (\bu , t)^2 \, \diff \bu \diff t.
\end{equation}

\begin{remark}
The test statistics $\hat T_n$ is of Cram\'er-von Mises type, as opposed to the Kolmogorov-Smirnov type (which would be defined taking the sup instead of integrating). In \cite{genest+r:2004} these two types of statistics are introduced in the context of (unconditional) independence testing. In our context, the Cram\'er-von Mises type is preferred over the Kolmogorov-Smirnov for practical reasons. Indeed, as we will see in the next section, a closed formula exists for  $\hat T_n$.
\end{remark}

\subsection{Computation of the statistic}
\label{sec:computation_stat}

The following lemma provides a closed-formula for the test statistics $\hat T_n$. The proof is left in the supplementary material.

\begin{lemma}\label{eq:statistic_T}
If $w: \mathbb R^d \to \mathbb R $ is an integrable function, then
$$
\hat T_n = n^{-2} \sum_{1\leq i,j \leq n } M\left( \hat{\bm {G}}_{i} , \hat{\bm {G}}_{j} \right)   w^\star (X_i - X_j)
$$
where $\hat{\bm {G}}_{i} = (  \hat{R}_{i1} - 1 ,  \hat{R}_{i2}  - 1 ) / n $, $w^\star = w \star w_s$ with $w_s (x) = w(-x)$, and
\begin{align*}
  &M(\mathbf{u}, \mathbf{v} ) = \\
  &\left( 1 - u_{1} \vee v_{1} \right) \left( 1 -u_{2}\vee v_{2}  \right) - 
  \frac{1}{4} \left\{ \left(1 -u_ {1} ^2 \right) \left(1 - u _ {2} ^2 \right) + \left(1 -v_ {1} ^2 \right) \left(1 -v_ {2} ^2 \right)\right\} + \frac{1}{9}  .
\end{align*}
\end{lemma}

\begin{remark}
The function $w$ is left unspecified for the sake of generality. Examples include $w(\cdot ) = \exp(-\|\cdot\|^2 )$, $w(t) = \mathds 1_{ \|\cdot \| \leq 1}$, where $\|\cdot\|$ stands for the Euclidean norm, and other popular kernel functions such as the Epanechnikov kernel. In the simulations, we consider the Gaussian kernel as in this case, $w^\star$ remains Gaussian. In line with the result stated in Proposition \ref{lemma:cara}, empirical evidences suggest that it does not have a leading role in the performance of the test.
%Another approach would have been to consider the (non-integrable) function $w(t) = \mathds 1_{\{X\leq t\}}$ as in \cite{stute:1997}. The same conclusion as in Lemma \ref{lemma:cara} remains valid in virtue of Proposition 16.10 in \cite{billingsley:1995}. 
\end{remark}

%\begin{remark}
%Computing $\hat T_n$ requires $n^2$ operations. A sampling strategy might be to rather compute
%\begin{align*}
% \frac{1}{|I\times J|} \sum_{(i,j) \in (I\times J) } M(\hat{\bm {R}}_{i} , \hat{\bm {R}}_{j})   w^\star (X_i - X_j) ,
%\end{align*}
%with $|I| =|J|$ denote random samples uniformly drawn in $\{1,\ldots, n\}$.
%The performance of the test strongly relies on good estimates of the conditional distributions $F_j(Y_j|X)$.
%\end{remark}

\subsection{Bootstrap approximation} 
\label{sec:bootstrap_approx}

To compute the rejection regions of the test, we propose a bootstrap approach to generate new samples in a way that reflects the null hypothesis even when $\mathcal H_0$ is not realized in the original sample. This has been notified as a guideline for bootstrap hypothesis testing in \citep{hall+w:1991} and it enables, in practice, to control for the level of the test and to obtain a sufficiently large power. 

The proposed bootstrap follows from the estimated conditional marginals of $Y_1|X$ and $Y_2|X$, respectively $\hat F_{n,1}$ and $\hat F_{n,2}$, and from the estimated distribution of $X$, denoted by $\hat F_{n}$. 
%Let $(X_i, Y_{i1}, Y_{i2})$, for $i \in \{1, \ldots, n\}$, be independent and identically distributed random vectors, with common distribution equal to the one of $(X, Y_1, Y_2)$. Based on this dataset, we first build the random functions $\hat{F}_{n,j}$ for $j \in \{1, 2\}$ as described in \ref{subsec:loc_lin_def}. 
First choose ${X}^*$ uniformly over the $(X_i)_{i=1,\ldots, n}$, that is, $X^* \sim  \hat F_{n}$. Then generate 
\begin{align*}
 & {Y_1}^* \sim \hat{F}_{n,1}(\cdot|X^* ),\qquad \text{and} \qquad {Y_2^*} \sim \hat{F}_{n,2}(\cdot|X^*).
\end{align*} 
Execute the previous steps $n$ times until obtaining an independent and identically distributed bootstrap sample of size $n$. Denote by $(X_i^*,Y_{i1}^*,Y_{i2}^*)_{i=1,\ldots , n} $ the obtained sample. Compute the test statistic based on this sample. We repeat this $B$ times and obtain $B$ realizations of the statistic under $\mathcal H_0$, denoted by $(T_{n,1}^{*}, \dots, T_{n,B}^{*})$. Now define the cumulative distribution function of the bootstrap statistics $t\mapsto ({1} / {B} )  \sum_{b=1}^B \ind_{ \{ T_{n,1}^{*}\leq t\} }$, and denote by $\xi_n (\alpha)$ its quantile of level $\alpha \in (0,1)$. The weighted partial copula test statistic with level $\alpha$ rejects $\mathcal H_0$ as soon as $ \hat T_n > \xi_n (\alpha)$.
%Then, one wants to approximate the distribution of the statistic $T_{1,n}$ under the null that can be mimic as follows. First \textit{i.i.d} random copies of $(\overline{X}, \overline{Y_1}, \overline{Y_2})$ are drawn as described above. Then $\hat{T}_{1,n}$ is estimated on this sample $\overline{Z}_n^{(b)}$. We repeat these two steps $B$ times and obtain $B$ realizations of the statistic under $\mathcal H_0$ denoted by $\left(\overline{t}^{(1)}, \dots, \overline{t}^{(B)}\right)$. We can then obtain the $p$ value as: 
%\begin{equation}
%    p = \frac{1}{B} \sum_{b=1}^B \ind_{t_Z \leq \overline{t}^{(b)}}. 
%\end{equation}

%Sampling under the null has been extensively studied for permutation-based test where the null-distribution is generated by computing the test statistic from a permuted sample. As example, in \cite{runge2017conditional}, the author combines a nearest-neighbor permutation that reproduces conditional independence and a conditional mutual information estimator to test conditional independence. In \cite{doran2014permutation}, the authors propose to use of a $k$-nearest-Neighbor permutation to transform a conditional independence test problem to a easier two-samples test problem. Hereinafter, unlike previous works, we do not make the use of any permutation to simulate the null. 

% \pb{J'imagine que l'idéal serait d'obtenir une garantie théorique sur cette procédure de bootstrap ?
% Je veux dire par là montrer que $P_{H_0}( \hat T_n > \xi_n (\alpha))$ tend vers $\alpha$.}

% \fp{Oui tout à fait}

\subsection{A generic example using nearest neighbors}
\label{sec:generic_example}

In this section, the aim is to illustrate the proposed test procedure when using the popular \textit{nearest neighbors} estimator for the margins $F_j$, $j\in \{1,2\}$.

\paragraph{Nearest neighbor estimator.}

%A crucial step of the proposed approach is the estimation of the margins $\hat{F}_{n,j}(y_j|x)$ for $j\in \{1,2\}$.  
Let $x\in S_X$ be fixed and $k\in\{1,\; \ldots,\; n\}$.   Let ${\hat B}_k(x) $ be the ball with center $x$ and smallest radius so that it  contains at least $k$ data points. The metric used here is the Euclidean distance but the theoretical results furnished in Section \ref{sec:results} remains valid whatever the distance.
For $j\in\{1,2\}$, the nearest neighbor estimator of $F_j(\cdot | x)$ is given by 
\begin{align}
\label{def:margins_estimators}
\hat F_{n,j}^{(NN)} (y|x) = \frac{ \sum_{i=1}^n \ind_{\{Y_{ij} \leq y\}} \mathbb I _{{\hat B}_{k_j} (x)  }(  X_i)}{\sum_{i=1}^n \mathbb I _{{\hat B}_{k_j}  (x)  }(  X_i)}\,,\qquad (y\in {\mathbb R}).
\end{align}
The rationale behind the nearest neighbor estimate is the one of local averaging as explained for instance in \cite{gyorfi2006distribution}. The choice of the number of neighbors $k_1$ and $k_2$ is discussed below.

\begin{figure*}[t]
    \centering
    \includegraphics[width=.85\textwidth]{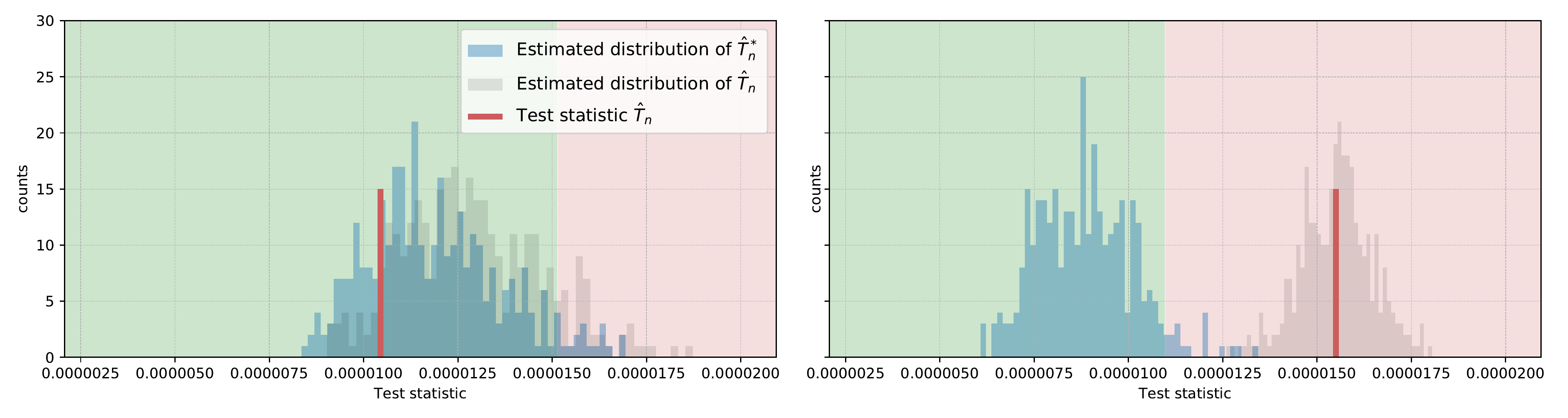}
    \caption{%\pb{Attention aux points d'interrogation dans la figure.} 
    The left (resp. right) figure corresponds to $a=0$, i.e., $\mathcal H_0$ is valid (resp. $a=0.5$, i.e., $\mathcal H_0$ is false). Based on a dataset of size $ n=1000$, $\hat T_n$ is computed (red line). Bootstrap statistics $(T^*_{n,b})_{b=1,\dots,B}$ $(B=300)$, are used to obtain the distribution of $\hat{T}^*_n$ (blue). The red area is the rejection region at $5\%$. Independently, $\hat{T}_n$ is computed $300$ times to obtain the distribution of in gray. 
%\pb{La légende du gris est ``Distribution of the bootstrap statistics'', on ne sait plus vraiment de quel bootstrap on parle...
%Dans l'ensemble, pour moi cette figure n'est pas très claire, et il manque la valeur de $a$, comment sont calculées les régions de rejet, à quoi correspond la valeur de 15 de la barre rouge, etc. Par ailleurs, Je ne dirais pas ``sous H0 en haut et sous H1 en bas'', mais plutôt ``avec $a=0$ en haut (H0 est vraie) et avec $a=\dots$ en bas (H1 est vraie)''.} 
     }
    \label{fig:test_illustration}
\end{figure*}
% Let $(X_i, Y_{i1}, Y_{i2})$, for $i \in \{1, \ldots, n\}$, be independent and identically distributed random vectors, with common distribution equal to the one of $(X, Y_1, Y_2)$.

\paragraph{Cross-validation selection of the number of neighbors.}
\label{sec:cross_val_kernel_bandwidth}
%\pb{On cross-valide sur la MSE, il n'y a pas vraiment de raison théorique de faire ça, mais effectivement je ne vois guère ce qu'on pourrait faire de mieux.} \fp{Oui tout à fait, et c'est rapide}
The number of neighbors $k_1$ and $k_2$ have a critical effect on the shape of the resulting estimates, and thus on the performance of our test. 
Indeed, these estimates of the margins $\hat{F}^{(NN)} _j(y_j | x)$ for $j \in \{1, 2\}$ are used in the computation of the test statistic $\hat{T}_n$ as well as in the bootstrap procedure to simulate under the null (see Section \ref{sec:bootstrap_approx}). %The effect on our test is then twofold. In the first place, it must enable our test statistic to properly discriminate $\mathcal{H}_0$ from $\mathcal{H}_1$. Second, it must enable to reproduce properly the behaviour of our proposed statistic under the null.
The idea is to assess the performance of each regression model $Y_1|X$ and $Y_2|X$ and to choose each $k_1$ and $k_2$ accordingly. We randomly divide the set of observations into $L$ groups of nearly equal size. These groups are denoted by $\{I_\ell\}_{\ell =1,\ldots, K}$. Define $\MSE_{j,\ell}(b) = ({1}/{|I_\ell|} ) \sum_{i\in I_\ell} (Y_{ij} - \hat{g}_{j,k}^{(-\ell)}(X_i))^2$,  where $\hat{g}_{j,k}^{(-\ell)}$ stands for the $k$-nearest neighbors estimate of the regression $Y_j | X$ computed on $\{1, \dots, n\}\backslash I_\ell$. We choose $k_{j}$ as the minimizer of $ ({1}/{L}) \sum_{\ell=1}^L \MSE_{j,\ell}(k)$ over $k$. 
 
 The success of the approach in distinguishing $\mathcal H_0$ from its contrary is illustrated on Figure \ref{fig:test_illustration} considering the generic post-nonlinear noise model (when $d= 1$) as described in \ref{sec:numerical}.

\begin{remark}
Though this example has been carried out using the nearest neighbors estimate of the marginal distributions, other approaches to estimate the marginals can be used to conduct the weighted partial copula test. The only restriction on the employed marginal estimates comes from the bootstrap procedure in which the ability to generate according to each margin is required. This is satisfied for most of the parametric models, e.g., Cox model or Gaussian regression. The broad class of Stone's weighted regression estimates \citep{stone1977consistent} can also be implemented in our procedure by using multinomial sampling to generate bootstrap samples. This class is important as it includes most of the popular nonparametric estimates: Nadaraya-Watson, nearest neighbors, local linear regression, smoothing splines and Gasser-Muller weights. We refer to \citep{gyorfi2006distribution} and \cite{fan1996} for an overview of these techniques.
\end{remark}

\section{Main results}
\label{sec:results}

%\subsection{Smooth estimator of the margins}
%\label{subsec:loc_lin_def}

\subsection{Consistency of the empirical weighted partial copula and of the test statistics}

%We rely on the following H\"older regularity class. Let $\delta\in (0, 1)$, $k\in \mathbb N$, and $M>0$ be scalars and let $S\subset \R $ be non-empty, open and convex.  Let $\mathcal C_{k+\delta,M}(S)$ be the space of functions $S\rightarrow \R$ that are $k$ times differentiable and whose derivatives (including the zero-th derivative, that is, the function itself) are uniformly bounded by $M$ and such that every mixed partial derivative of order $l \leq k$, say $f^{(l)}$, satisfies the H\"{o}lder condition 
%\begin{align}\label{eq:Holder}
% \sup_{z\neq \tilde z}\,  \frac{\abs{ f^{(l)}(z)-f^{(l)}(\tilde z) }}{\abs{z-\tilde z}^{\delta}}   \leq M ,
%\end{align}
%where $\lvert\point\rvert$ in the denominator denotes the Euclidean norm. In particular, $\mathcal C_{1,M}(\R)$ is the space of Lipschitz  functions $\R \rightarrow \R $ bounded by $M$ and with Lipschitz constant bounded by $M$.

The results of the section are provided in a general setting in which the estimates of the conditional margins are left unspecified but should satisfy a particular assumption, namely the \textit{uniform consistency} (UC), given below. We give some examples after the statements of the main results. Let $\mathbb P$ denote the probability measure on the underlying probability space associated to the whole sequence $(X_i,\bY_i)_{i = 1,2,\ldots}$. In what follows $(r_n)_{n\geq 1} $ denotes a positive sequence. We write $Z_n = O_{\mathbb P}(r_n)$, if $Z_n / r_n$ is a tight sequence of random variables.

\begin{enumerate} 
\item[(UC)] \label{cond:high_level_consistency} For any $j\in\{1,2\}$, we have 
$$ \sup_{x\in S_X ,\, y_j\in \mathbb R} |\hat F_{n,j}(y_j | x) - F_j(y_j|x) | =O_{\mathbb P} (r_n)  $$
\end{enumerate}

The following result is interested in the uniform consistency of the weighted partial copula process. Its proof is given in the supplementary material.

\begin{theorem}
\label{theorem:consistency}
Assume that (UC) holds true and that $w(\cdot)  = \tilde w ( \|\cdot \|) $, where  $\tilde w: \mathbb R_{\geq 0}  \mapsto \mathbb R$  is of bounded variation. We have, when $n \to \infty$,
\begin{align*}
\sup_{ \bu \in  [0,1 ] ^2 ,\, t\in \mathbb R } \left| \hat W_n (\bu, t)  -    W (\bu , t)  \right| = O_{\mathbb P} ( n^{-1/2} + r_n ).
\end{align*}
%In addition,  the process $\left\{  n^{1/2} \hat W_n(\bu,w)\right\}_{\bu\in [\gamma,1-\gamma]^2,t\in \mathbb R}  $ converges weakly in $\ell^\infty( [\gamma, 1-\gamma] ^2 \times \R) $ to a certain Gaussian process.
\end{theorem}

% Let $\ell^\infty(T)$ denote the space of bounded real functions on the set $T$, the space being equipped with the supremum distance. Define $U_{i1} = F_1(Y_{i1}|X_i) $, $U_{2i} = F_2(Y_{i2}|X_i)$, for any $i=1,\ldots, n$, and  
%\begin{equation}
%\tilde W_n (\bu , t) = \hat Z_n (\bu,t) -  (f_X \star w)(t) ( u_1\hat Z_{n,2} (u_2) + u_2\hat Z_{n,1} (u_1) ), 
%\end{equation}
%for any $\bu \in [0, 1]^2$, $ t\in\R$, with
%\begin{equation}
%\hat Z_n (\bu,t) =
%\quad n^{-1} \sum_{i=1} ^n \Big\{ w(t-X_i) ( \mathds 1 _{\{  U_{i1} \leq u_1,\,  U_{i2}  \leq u_2\} } -  u_1u_2 )   \Big\}, 
%\end{equation}
%and  $ \hat Z_{n,j} (u_j) =  n^{-1} \sum_{i=1} ^n \left \{   \mathds 1 _{\{  U_{ij} \leq u_j\} } - u_j  \right\} $. 

%\begin{theorem}
%\label{theorem:weakcv}
%Assume that (G\ref{cond:smoothnessdensity1}), (G\ref{cond:class_W}), (H\ref{cond:high_level_inverse_inequality}), (H\ref{cond:high_level_consistency}) and (H\ref{cond:high_level_donsker}) hold. If $\mathcal H_0$ holds, then for any $\gamma\in (0,1/2)$, we have when $n \to \infty$
%\begin{align*}
%\sup_{ \bu \in  [\gamma,1 - \gamma] ^2 ,\, t\in \mathbb R } \left| \hat W_n (\bu, t)  -    \tilde W_n (\bu , t)  \right| = o_{\mathbb P} (n^{-1/2}).
%\end{align*}
%In addition,  the process $\left\{  n^{1/2} \hat W_n(\bu,w)\right\}_{\bu\in [\gamma,1-\gamma]^2,t\in \mathbb R}  $ converges weakly in $\ell^\infty( [\gamma, 1-\gamma] ^2 \times \R) $ to a certain Gaussian process.
%\end{theorem}

\begin{remark}
Theorem \ref{theorem:consistency} is an extension of Theorem 1 in  \cite{omelka+g+v:15}. By taking $w = 1$ we would recover their result.
\end{remark}

\begin{remark}\label{rk:proof3}
The rate of convergence obtained in Theorem \ref{theorem:consistency} for the partial copula process is as bad as the rate of convergence of the marginal estimate. When the marginal estimates are parametric and therefore attached to a $1/\sqrt n $-rate our bound might be sharp. However, when these estimates are nonparametric we conjecture that an improvement in the rate of convergence can be obtained. To illustrate this claim, in the case when $w = 1$, $d=1$, and using a smoothed kernel estimate for the margins, \cite{portier+s:2018} obtained a rate of convergence of $1 / \sqrt n $, which is faster than the usual nonparametric convergence rate. Known approaches to achieve such a program relies on the complexity of certain classes generated by conditional quantile estimates. Obtaining guarantees on such complexity would require to focus on particular estimates as the one used in \cite{portier+s:2018}. 
\end{remark}

We now provide the following result illustrating the antithetic behavior of $\hat T_n$ under $\mathcal H_0$ and its contrary. This cannot be obtained as a corollary of the previous result due to the integration over $t$ which is carried-out over an infinite domain.

\begin{theorem}\label{theorem:consistency2}
Assume that (UC) holds true and that $w: \mathbb R^d \to \mathbb R $ is integrable with respect to the Lebesgue measure and with a Fourier transform being non-zero almost everywhere. When $\mathcal H_0$ is true, we have that $  \hat T_n   = O_{\mathbb P} (n^{-1} + r_n^2) $ as $n \to \infty$.  When $\mathcal H_0$ is false, there exists $\epsilon>0$ such that $ \hat T_n  > \epsilon  $ with probability going to $1$ as $n \to \infty$
\end{theorem}

\begin{remark}
The extension to higher dimension $\bY = (Y_1,\ldots, Y_p)$ is straightforward and only represents some minor changes in the proofs of \ref{theorem:consistency} and \ref{theorem:consistency2}. If the rate $r_n$ is not the same for each margin, the final rate obtained is the largest among them.
\end{remark}

\begin{remark}\label{rem:exemple}
Assumption (UC) is satisfied for usual parametric classes $\{F_\theta \, :\, \theta\in \Theta\}$.  Its validity follows from the consistency of the parameter estimate along with the continuity property $\sup_{x,y}| F_{\theta_0}(y|x)  - F _{\theta}(y|x)  |\to 0$ as soon as $\theta\to\theta_0$.  Concerning nonparametric methods, (UC) has been the object of several studies for different regression estimates: nearest neighbors \citep{jiang2019non}, Kaplan-Meier \citep{dabrowska1989uniform} and Nadaraya-Watson \citep{hansen2008uniform}. For Nadaraya-Watson weights, (UC) is obtained in \cite[Corollary 2]{einmahl2000empirical} and for a smoothed version of Nadaraya-Watson (UC) is investigated in \cite{portier+s:2018}. In the next section, we obtain UC for the nearest neighbor estimate described in Section \ref{sec:def}.
\end{remark}

\subsection{Uniform consistency of nearest neighbors marginals estimates}

The result presented in this section ensures that (UC) is satisfied for the nearest neighbor estimate  of the margins \eqref{def:margins_estimators}. This uniform convergence result, which is new to the best of our knowledge, is introduced independently from the rest of the paper. Let $(X_i, Y_{i} )$, for $i \in \{1, \ldots, n\}$, be independent and identically distributed random vectors, with common distribution equal to the one of $(X, Y)$. The distribution of $X$, $f_X$ is supposed to satisfy the following assumption. 

\begin{enumerate} 
\item[(B)] The density $f_X$ is bounded and bounded away from $0$ on $S_X = [0,1]^d$.
\end{enumerate}

The assumption that $S_X$ is the unit cube can be alleviated at the price of some technical assumption given in \citep[Assumption 1]{jiang2019non}.
Denote by $F(\cdot|x)$ the conditional distribution of $Y$ given $X=x$. We require the following regularity condition.

\begin{enumerate} 
\item[(L)] The function $x\mapsto F(y|x)$ is Lipschitz uniformly in $y\in \mathbb R$, i.e.,
\begin{align*}
\forall (x, x') \in S_X \times S_X,\qquad  \sup_{y\in \mathbb R} |F(y|x)- F(y|x') | \leq  L \|x-x'\|.
\end{align*}
\end{enumerate}

\begin{theorem}\label{consistenc:kNN}
Under (B) and (L), if $k/n \to 0 $ and $\log(n) / k \to 0$, we have almost-surely,
\begin{align*}
\sup_{y\in \mathbb R, \, x\in S_X} | \hat F^{(NN)}  (y|x)  - F(y|x)  |  = O  \left(  \sqrt { \frac{ d}{k}  \log( n )}  + \left( \frac{ k }{ n} \right) ^{1/d} \right) ,
\end{align*}
where $\hat F^{(NN)} $ is defined in \eqref{def:margins_estimators}.
\end{theorem}

\begin{remark}
The simple bound given in the previous result is valid in the asymptotic regime $n\to \infty$. It could however be extended to finite sample results at the price of additional technical details.
%The simple bound given in the previous result (as well as those given in the previous section) are valid in the asymptotic regime $n\to \infty$. They could however be extended to finite sample results at the price of additional technical details.
\end{remark}

\begin{remark}
The approach employed in the proof relies on controlling the complexity of the nearest neighbors selection mechanism. The main idea is to embed the nearest neighbors selector in a nice class of kernel function having a control complexity. The technique appears to be new in this context and general enough to be extended to estimate of the type
$$
k^{-1} { \sum_{i = 1}^n g (Y _i )  \mathbb I _{{\hat B}_{k} (x)  }(  X_i)} ,
$$
when $g$ is lying on a certain VC class of function. Similar quantities, using a kernel smoothing approach, have been considered in \citep{einmahl2000empirical}.
\end{remark}

\begin{remark}
Optimal rates of convergence (for Lipschitz functions) are achieved when selecting $  k  =   n^{2 / (d+2) }   $ as it gives a uniform error of order $ n ^{ -  1 / (d+2) }  $ \citep{stone1982optimal}. 
\end{remark}

%\begin{assumption}\label{cond:lip3}
%The function $f_X$ is $L$-Lipschitz at $x$ on $\mathcal{B}(x, \tau_0)$, \textit{i.e.} there exists $L>0$ such that for all $z \in B(x, \tau_0)$,
%\begin{align*}
%|f_X (z) - f_X(x) | \leq  L\|z-x\| .
%\end{align*}
%\end{assumption}

%%% Local Variables: 
%%% mode: latex
%%% TeX-master: "../main"
%%% End: 

\section{Numerical experiments}\label{sec:numerical}
In this section, we apply the proposed copula test to synthetic and real data to evaluate its performance based on the nominal level and the power of the test. 
The weighted partial copula test is put to work with two different margins estimates. As a first approach, we consider the following parametric estimate of the conditional margins: 
\begin{equation}
\hat F_{n,j}^{(LR)} (y|x) = \Phi_{\left( x^\top  \hat \beta_j, \hat \sigma \right)} (y),\qquad (y\in {\mathbb R}),
\end{equation}
where $\Phi_{(m,\sigma )}$ denotes the c.d.f. of a normal distribution of mean $m$ and mean $m$ and variance $\sigma^2$. This estimate will be referred to as the linear regression (LR) estimate because  $\hat \beta_j$ is estimated by minimizing the sum of squares in a linear model ($\hat \sigma^2 $ is chosen as the variance of the estimated residuals). In fact, $\hat F_{n,j}^{(LR)}$ is the classical Gaussian maximum likelihood estimator of the margin $F_j$. Representing a more flexible approach, we use the nearest neighbors (NN) estimate given in \eqref{def:margins_estimators} for the margins.  For both approach, the function $w$ is a Gaussian kernel given by $w(t) = \exp(-\|t\|^2)$ and $B=200$ bootstrap realizations.

We compare our approach with the KCI-test \cite{zhang2012kernel} presented in the related literature section.
Since the level $\alpha$ is hard to set for the CCI-test of \cite{sen2017model}, this approach will only be considered when the proportions of correct decision will be computed (see Figure \ref{fig:p_correct_vs_a}).
The Python source code used to run our test is available in a supplementary material file.
%here\footnote{\url{https://anonymous.4open.science/r/07fbfb98-86e6-4924-b1a5-a024e8e55774/}}.

In the whole numerical study, we set $\alpha = 5\%$ as the desired type-I error rate. All results are averaged over $300$ trials.

\subsection{Linear Model}
\label{sec:linear}
Consider the joint distribution given by $Y_1=  X^T \beta_1 + \epsilon_1, Y_2= X^T \beta_2 + \epsilon_2,$ where $X \sim \mathcal{N}(0,\mathsf{I}_{d}), \beta_1 \text{and } \beta_2$ are two constant vectors of $[0,1]^d, $ and $\epsilon_1,\epsilon_2$ are two standard Gaussian variables with $\Cov(\epsilon_1,\epsilon_2)=a$.
When $a= 0$, $\mathcal{H}_0$ is true. It is false when $a >0$. 
We examine the effect of the constant $a>0$, and the size of the dataset $n$ on the type-{\MakeUppercase{\romannumeral 2}} error rate. We also examine the type-{\MakeUppercase{\romannumeral 1}} errors when the dimension of the variable $X$ increases, in a setting where the null hypothesis $\mathcal{H}_0$ holds.

%Two illustrative samples are shown in Figure \ref{fig:simu_under_linear}, one is drawn from the previous model with $a= 0$ and one other with $a= 0.5$. We set $\beta_1=\beta_2=(1/\sqrt{2}, 1\sqrt{2})^T \in \R^2.$ 
% We examine the effect of the constant $a>0$ on the area under the ROC curve (ROC AUC score) for the three tests in competition as well as on the type-{\MakeUppercase{\romannumeral 2}} error rate. 
%The ROC AUC score is estimated based on 200 experiments (100 under $H_0$, 100 under the alternative). For every experiment, the test statistics is computed based on 500 independent realizations of $(X_1,X_2,Y)$.
Figure \ref{fig:linear_noise_res} shows the attractive performance of our test compared to the KCI-test. Notably, we can see that in high dimensions, our test is more accurate with respect to the level set $\alpha$ than the KCI-test.
%Our test shows comparable ROC AUC scores to those obtained by KCI-test and outperforms the CCI-test.

% \begin{figure}[t]
% \centering
% %\includegraphics[width=.6\columnwidth]{roc_auc_vs_a_linear_model.pdf}
% \caption{The ROC AUC score computed for the three test in study, when the parameter $a$ varies. Here, $\beta_1=\beta_2=(1/\sqrt{2}, 1\sqrt{2})^T \in \R^2,$ and $n=500.$} \label{fig:roc_auc_vs_a}
% \end{figure}
\begin{figure*}[t]
\centering
\subfigure[\label{fig:type_II_error_vs_a_lm}$n=1000,$ various $a$,  $d=1$]{\includegraphics[width=0.3\textwidth]{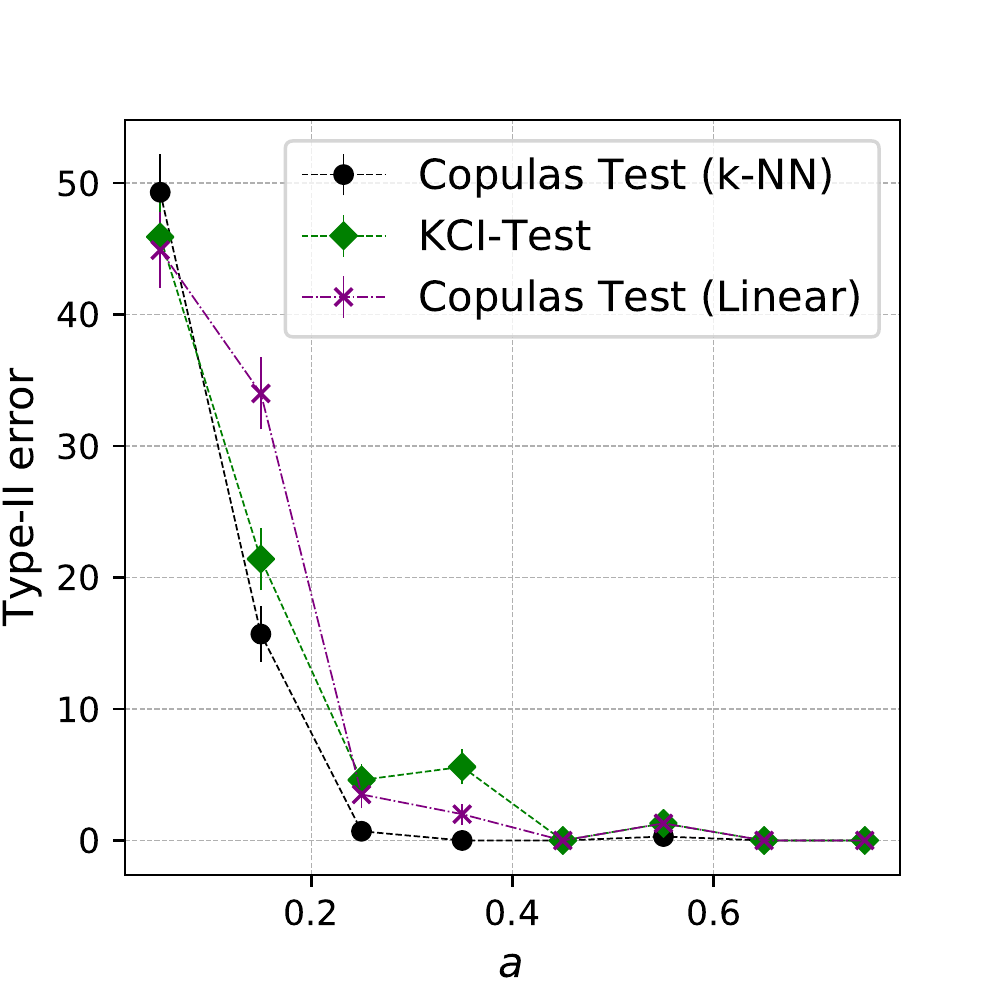}}
\subfigure[\label{fig:type_II_error_vs_n_lm} various  $n$, $a=0.3$, $d=1$]{\includegraphics[width=0.3\textwidth]{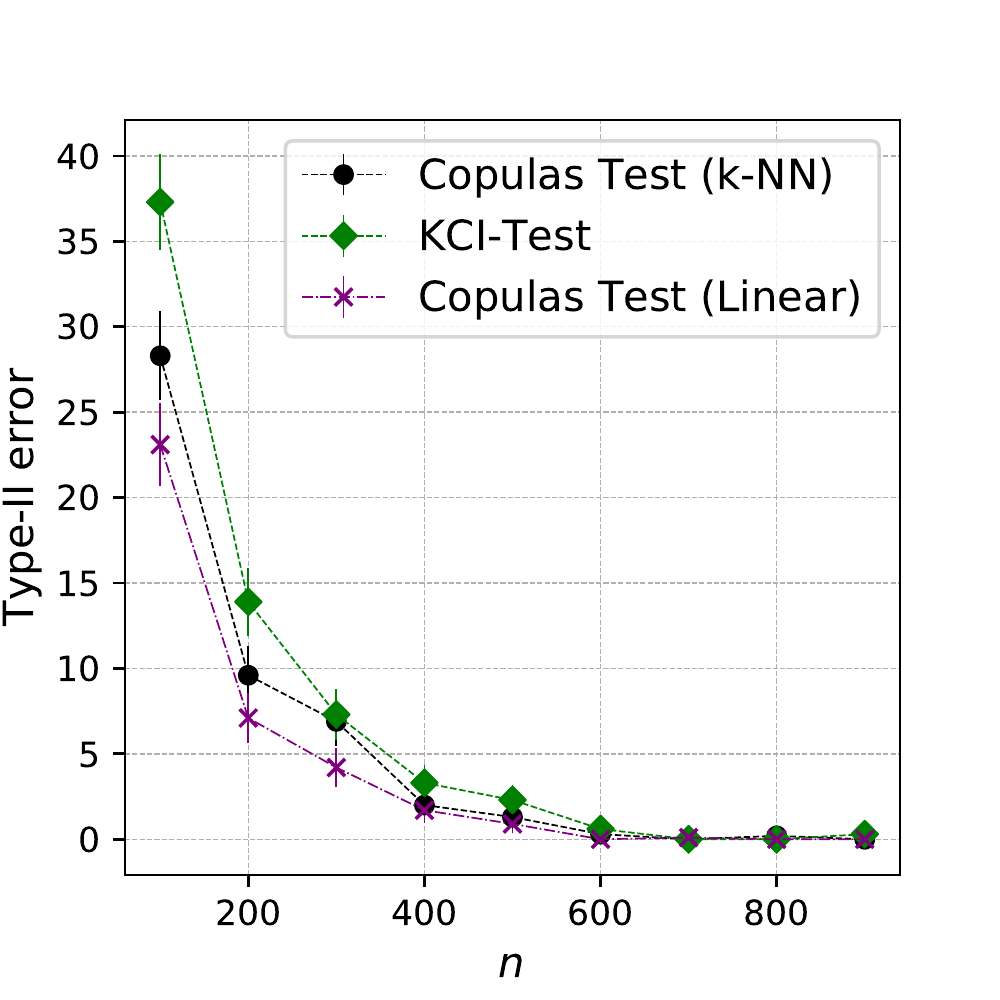}}
\subfigure[\label{fig:type_I_error_vs_d}$a=0$, $n=1000$, various $d$]{\includegraphics[width=0.3\textwidth]{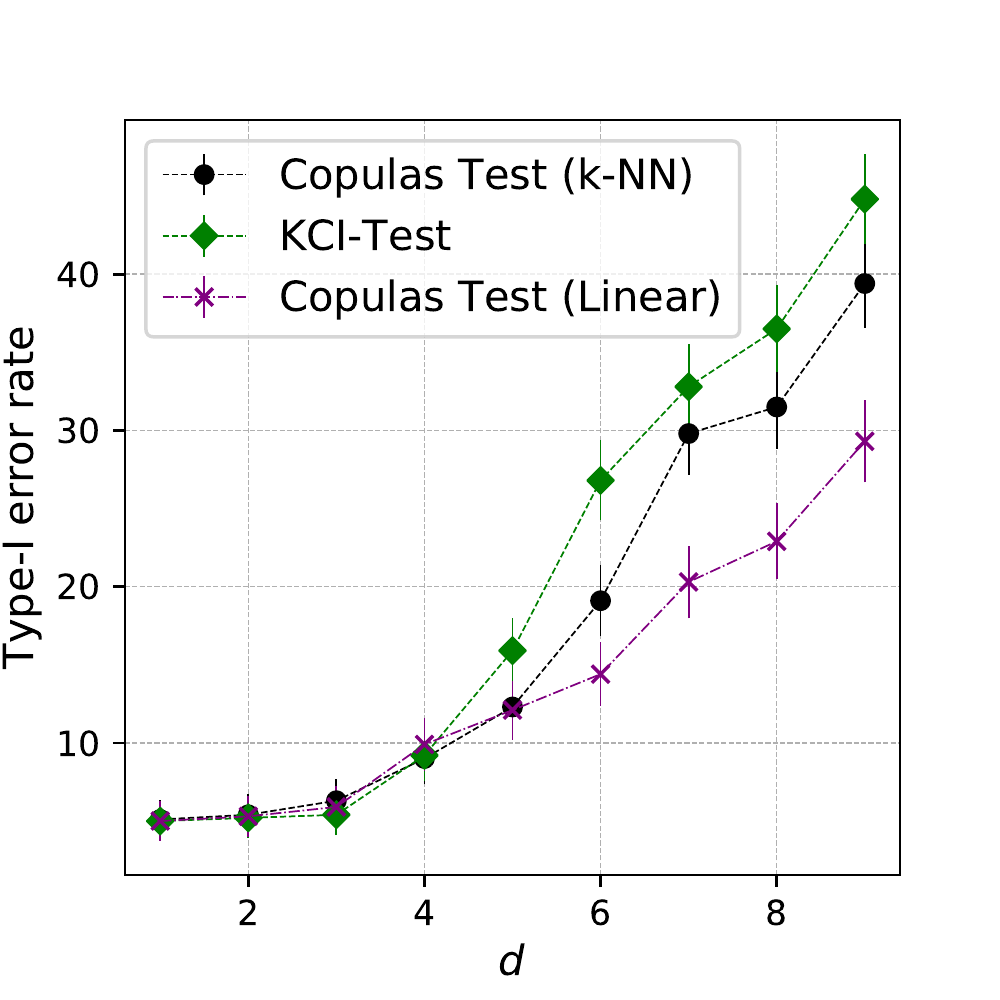}}
\caption{Simulation results for the linear model. Figures \ref{fig:type_II_error_vs_a_lm}, \ref{fig:type_II_error_vs_n_lm} show the probability of acceptance (i.e. the type II error
rate), plotted against the constant $a$ and $n$. Figure \ref{fig:type_I_error_vs_d} shows the probability of rejection (type I error) against $d$. The plots show the average probabilities with standard error bars.}
\label{fig:linear_noise_res}
\end{figure*}

\subsection{Disturbed Linear Model}
\label{sec:dist-linear}
In this section, we consider a slight perturbation appearing in the linear model of Section \ref{sec:linear} as follows: 
$Y_1 =  X^T \beta_1 + \epsilon_1$, $Y_2= c \|  X\| + X^T \beta_2 + \epsilon_2,$ where $X \sim \mathcal{N}(0,\mathsf{I}_{d}), \beta_1 \text{and } \beta_2$ are two constant vectors of $[0,1]^d, $ $c \geq 0$,  and $\epsilon_1,\epsilon_2$ are two centered Gaussian variables with $\Cov(\epsilon_1,\epsilon_2)=a$. 
We set $a=0.2$ and we examine the effect of the constant $c>0$, and the dimension $d$ of $X$ on the type-{\MakeUppercase{\romannumeral 2}} error rate. We also examine the type-{\MakeUppercase{\romannumeral 1}} errors when $d$ increases, in a setting where the null hypothesis $\mathcal{H}_0$ holds.

\subsection{Causality Discovery}
% Here, the joint distribution of $(X_1,X_2,Y)$ under the null is given by a Directed Acyclic Graph (DAG) $\mathcal{G} = (\mathcal{V}, \mathcal{E}),$ where $\mathcal{V} = (X, Y_1, Y_2)$ and $\mathcal{E}$ is random. 
Hereinafter we consider a particular type of DAG called latent cause model. 
%Indeed we have for the Markov chain: 
%\begin{align*}
%p(y_2|x, y_1) = %\frac{p(x, y_2, y_2)}{p(x, y_1)}= 
%\frac{p(x)p(y_1|x)p(y_2|y_1)}{\sum_{y_2'} p(x)p(y_1|x)p(y_2'|y_1)} 
%= p(y_2|y_1),
%\end{align*} hence the conditional independence $Y_1 \indep %Y_2~|~X$.
%For the DAG given in \ref{fig:graph_latent_cause}, we show that the conditional independence is true as well:
%\begin{align*}
%p(y_1, y_2|x) = %\frac{p(x, y_2, y_2)}{p(x)}=
%\frac{p(x)p(y_1|x)p(y_2|x)}{p(x)} = p(y_1|x) p(y_2|x).
%\end{align*}
%   \begin{figure}[t]
%       \centering
%       \subfigure[CI holds]{
%       \begin{tikzpicture}[auto,vertex/.style={draw,circle}]
%     \node[vertex] (a) {$X$};
%     \node[vertex, below left=.5cm of a] (b) {$Y_1$};
%     \node[vertex, below right=.5cm of a] (c) {$Y_2$};
%     \path[->]
%       (a) edge (b) 
%       (a) edge (c);
%         \label{fig:graph_markov_chain}
%   \end{tikzpicture}} \qquad \subfigure[CI fails]{
%   \begin{tikzpicture}[auto,vertex/.style={draw,circle}]
%     \node[vertex] (a) {$X$};
%     \node[vertex, below left=.5cm of a] (b) {$Y_1$};
%     \node[vertex, below right=.5cm of a] (c) {$Y_2$};
%     \path[->]
%       (a) edge (b) 
%       (a) edge (c);
%     \draw[->, dashed]
%         (b) edge (c);
%     \label{fig:graph_latent_cause}
%   \end{tikzpicture}}
%       \caption{Latent cause model. The CI holds when the dashed edge does not exist and fails otherwise.}
%       \label{fig:markov_models}
%   \end{figure}
To draw samples from the alternative hypothesis, we break the conditional independence by adding an edge between the nodes $Y_1$ and $Y_2$. 
%The resulting graphs are shown in Figure \ref{fig:markov_models} in dashed lines. 
For the ''Latent cause`` model of interest we have $X \sim \mathcal{N}(0, 1),$ $Y_1|X \sim \mathcal{N}(X, 1),$ and $Y_2|X, Y_1 \sim  \mathcal{N}(X + a Y_1, 1)$. It is easy to verify that $\mathcal{H}_0$ is true when $a=0$, and false otherwise. 
%For this experiment, we set the value of $a=0.2$. 
%For each of the type of model, we simulated 300 DAGs (150 under the null, 150 under the alternative, in which case we set $a=0.1$). For every experiment, the test statistics is computed based on $n$ independent realizations of the r.v., where $n$ is chosen between $200$ and $1000$. 
%Both our test and the KCI test are set to a significance level of $0.05$, while the threshold used in CCI-test is set to $0.5 - \frac{1}{\sqrt{n}}$ as recommended in the related paper \cite{sen2017model}. 
%Figure \ref{fig:prob_exact_mm} shows how often the tests in take the good decision. 
It can be seen in Figure \ref{fig:causal_inference_results} that for large sample size $n$, our test outperforms the ones in competition. Furthermore, our test is slightly more powerful than the KCI-test across a range of values of $a$ but overall shows fairly similar performance.
\begin{figure*}[!t]
\centering
\subfigure[\label{fig:p_correct_vs_a}$n=1500$, various $a$]{\includegraphics[width=0.3\textwidth]{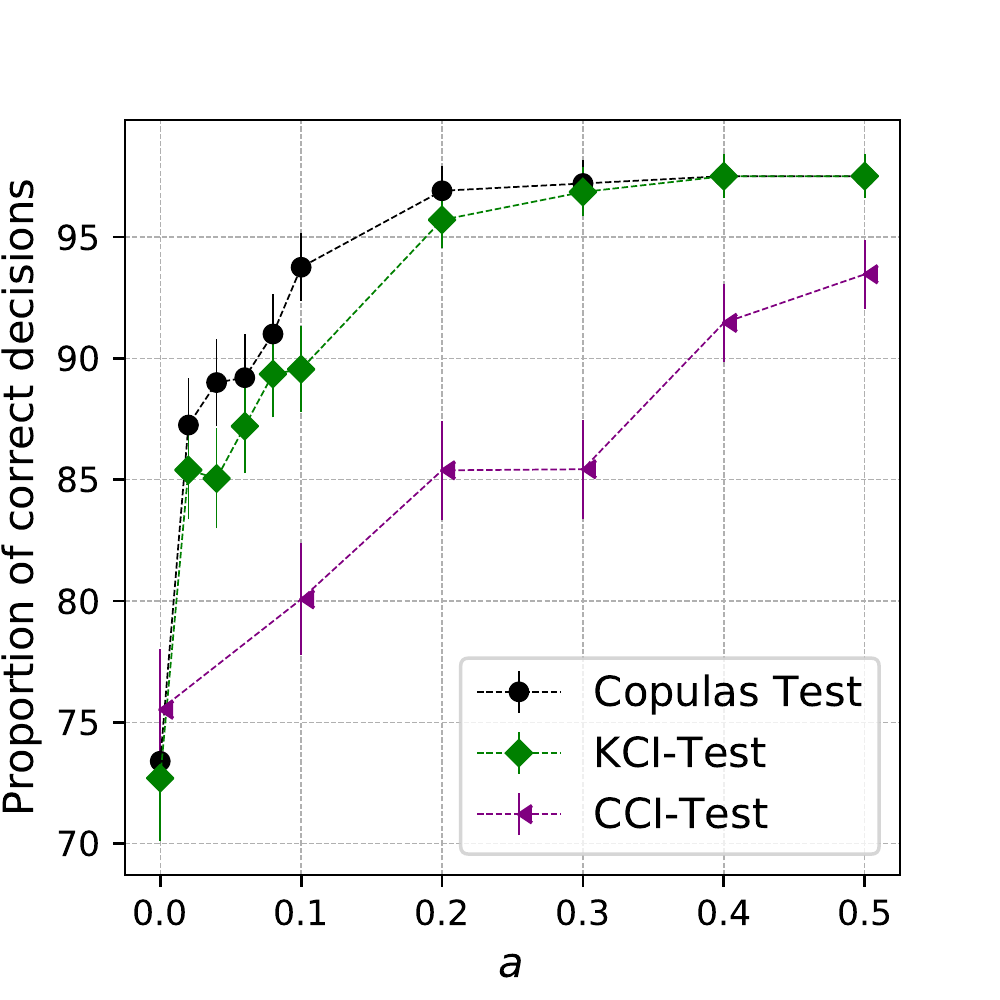}}
\subfigure[\label{fig:type_II_error_vs_a}$n=1500$, various $a$]{\includegraphics[width=0.3\textwidth]{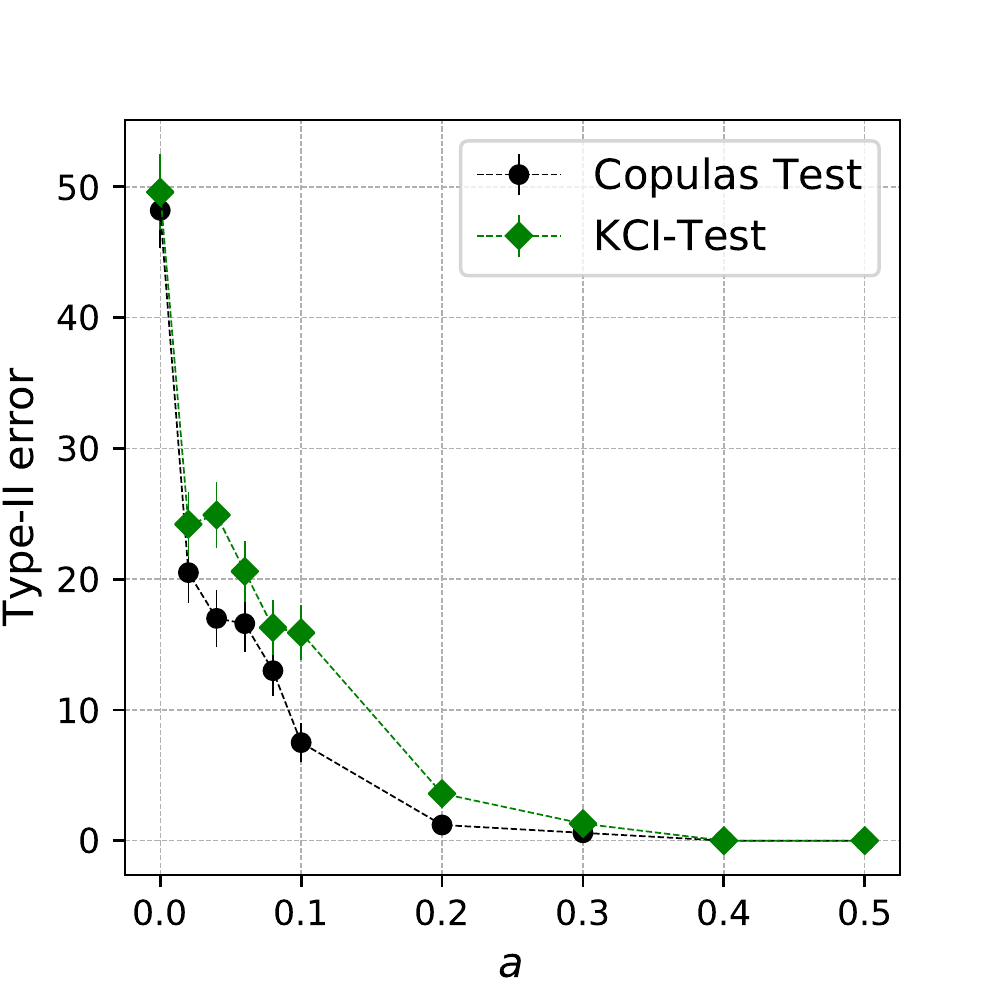}}
\subfigure[\label{fig:type_II_error_vs_n} $a=0.3$, various $n$]{\includegraphics[width=0.3\textwidth]{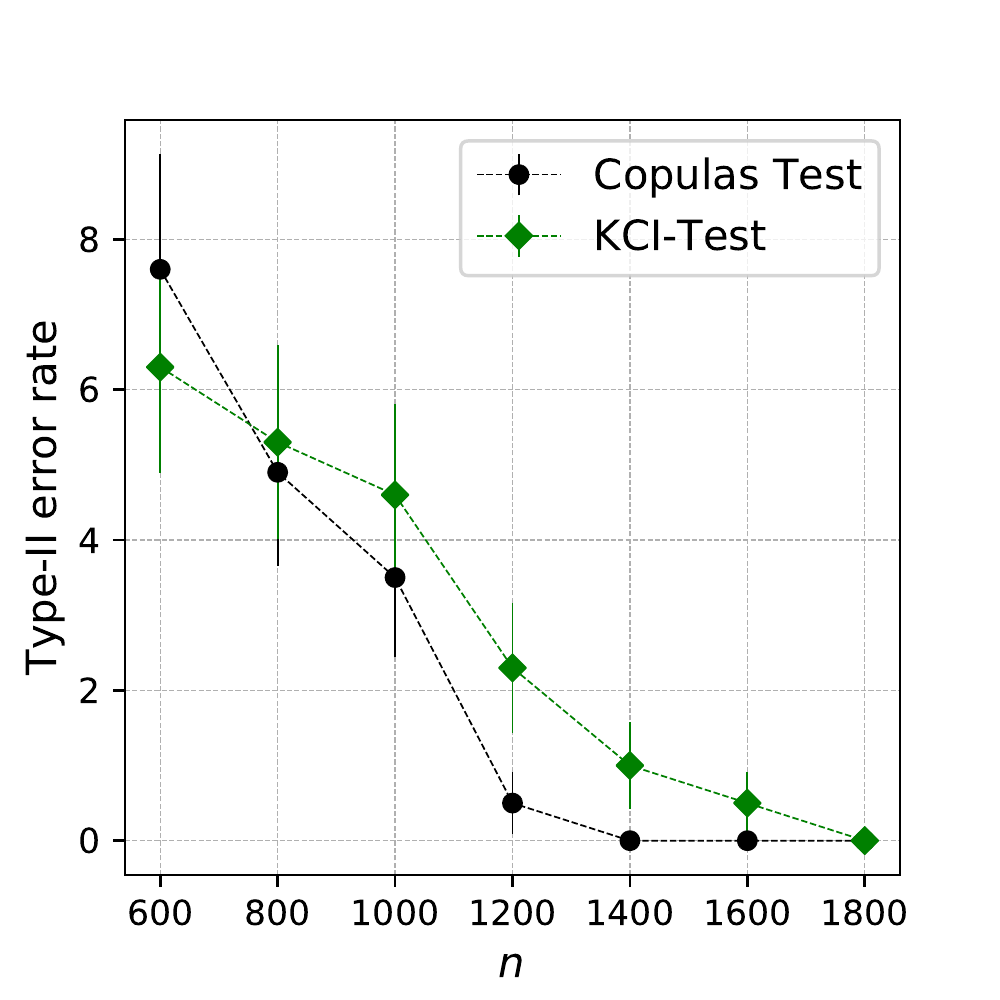}}
\caption{Simulation results for the causal inference datasets. }
\label{fig:causal_inference_results}
\end{figure*}

\subsection{Post-Nonlinear Noise}
\label{sec:post_nonlinear_noise}

We apply the proposed test on the post-nonlinear noise model.
%described in Section \ref{sec:generic_example}.
We first examine the effect of the constant $a>0$ on the probability of type-{\MakeUppercase{\romannumeral 1}} and  type-{\MakeUppercase{\romannumeral 2}} error of our test. 
%A good test is expected to have type-{\MakeUppercase{\romannumeral 1}} error closed to $\alpha$ and a small probability of type-{\MakeUppercase{\romannumeral 2}} error. Since the level $\alpha$ is, in practice, hard to set for the CCI-test, this latter does not appear in the figures.
%For each value of $a$, we simulate $600$ (300 under the null, 300 under the alternative). We then assess the performance of the test and compare it to the one obtained by the KCI-test. 
The results are shown in Figure \ref{fig:errors_vs_a_pnl}. As expected, larger values of $a$ yields lower type-{\MakeUppercase{\romannumeral 2}} error probabilities. For every value $a$, we observe that the type-{\MakeUppercase{\romannumeral 1}} error probability is close to $\alpha$. 
The performances of the tests are also compared when the sample size $n$ changes. The role of $n$ is critical and the results are shown in Figure \ref{fig:errors_vs_n_pnl}. We note that the  type-{\MakeUppercase{\romannumeral 1}} error probability is again close to $\alpha$ and that the type-{\MakeUppercase{\romannumeral 2}} quickly vanishes when $n$ increases. 
In this experiment, the proposed procedure outperforms the KCI-test.
%The poor performance of the test for low values of $n$ was expected, since it jeopardizes the quality of the margins estimates. 
%For a small sample size, the regression criterion proposed in Section \ref{sec:computation_stat} would favor large values for kernel bandwidth $b$. For such bandwidths, we observe (see Figure \ref{fig:performance_test_h}) that the type-{\MakeUppercase{\romannumeral 2}} grows critically. 

\begin{figure*}
\centering
\subfigure[\label{fig:roc_auc_vs_a_pnl}$n=500,$ various $a$ and  $d=1$]{\includegraphics[width=0.3\textwidth]{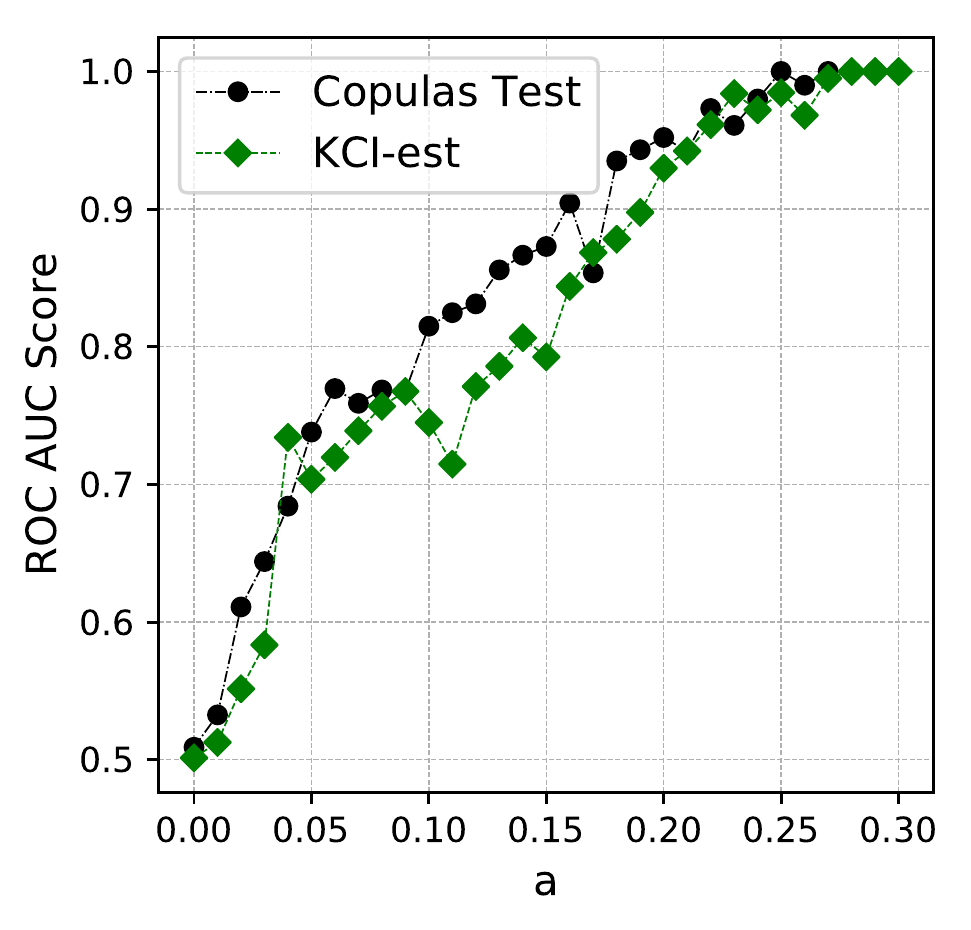}}
\subfigure[\label{fig:type_II_error_vs_a_pnl} $n=500,$ various $a$ and  $d=1$]{\includegraphics[width=0.3\textwidth]{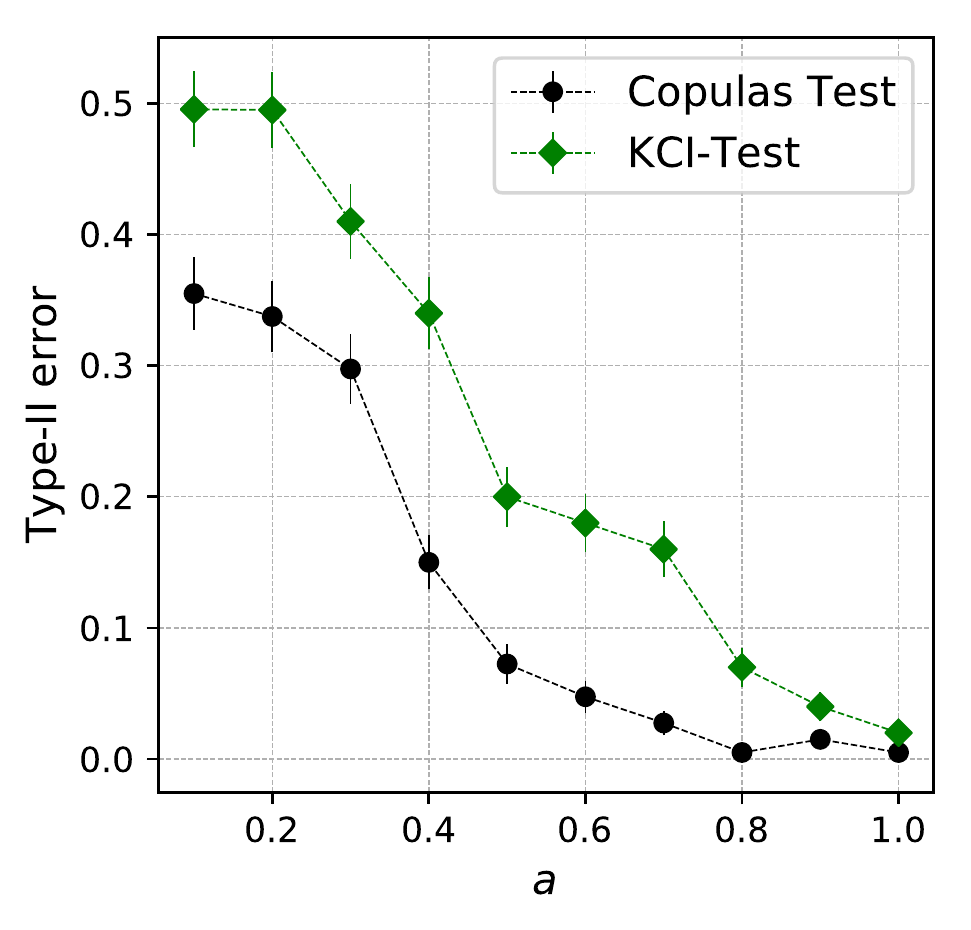}}
\subfigure[\label{fig:type_I_error_vs_a_pnl}$n=500,$ various $a$ and  $d=1$]{\includegraphics[width=0.3\textwidth]{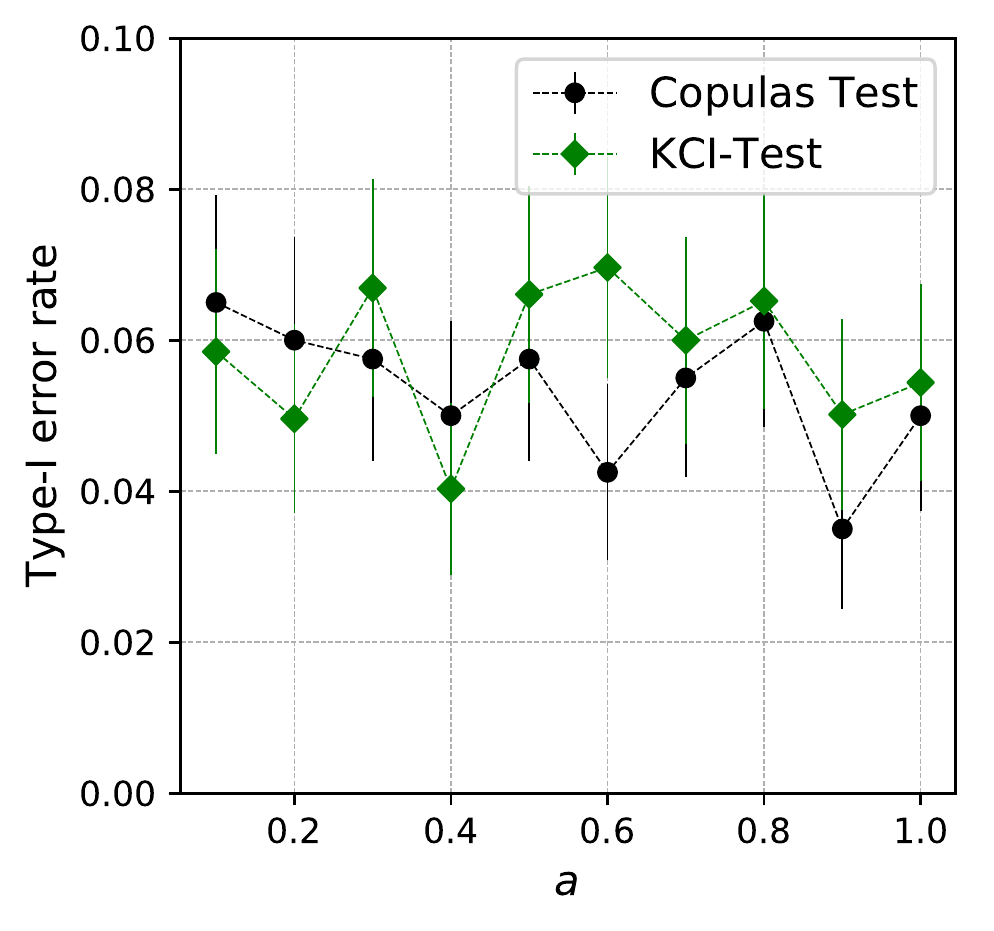}}
\caption{Simulation results for the post-nonlinear noise model. Figures \ref{fig:roc_auc_vs_a_pnl}, \ref{fig:type_II_error_vs_a_pnl} and \ref{fig:type_I_error_vs_a_pnl}, show respectively the ROC AUC score, the probability of acceptance (i.e. the type II error
rate), and the type I error
rate plotted against the constant $a$ with standard error bars.}
\label{fig:errors_vs_a_pnl}
\end{figure*}

\begin{figure*}[t]
\centering
\subfigure[$a=0.5$, various $n$, $d=1 $\label{fig:type_I_error_vs_n_pnl}]{\includegraphics[width=0.3\textwidth]{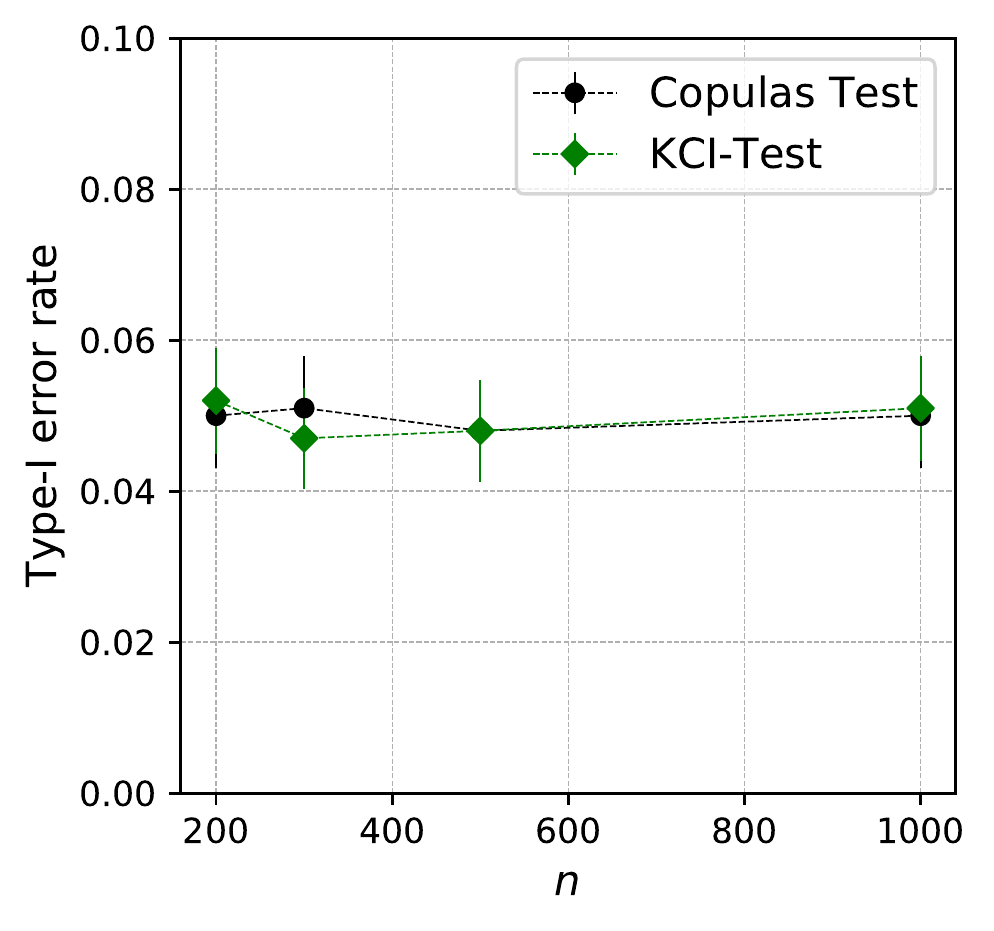}}
\subfigure[$a=0.5$, various $n$, $d=1 $\label{fig:type_II_error_vs_n_pnl}]{\includegraphics[width=0.3\textwidth]{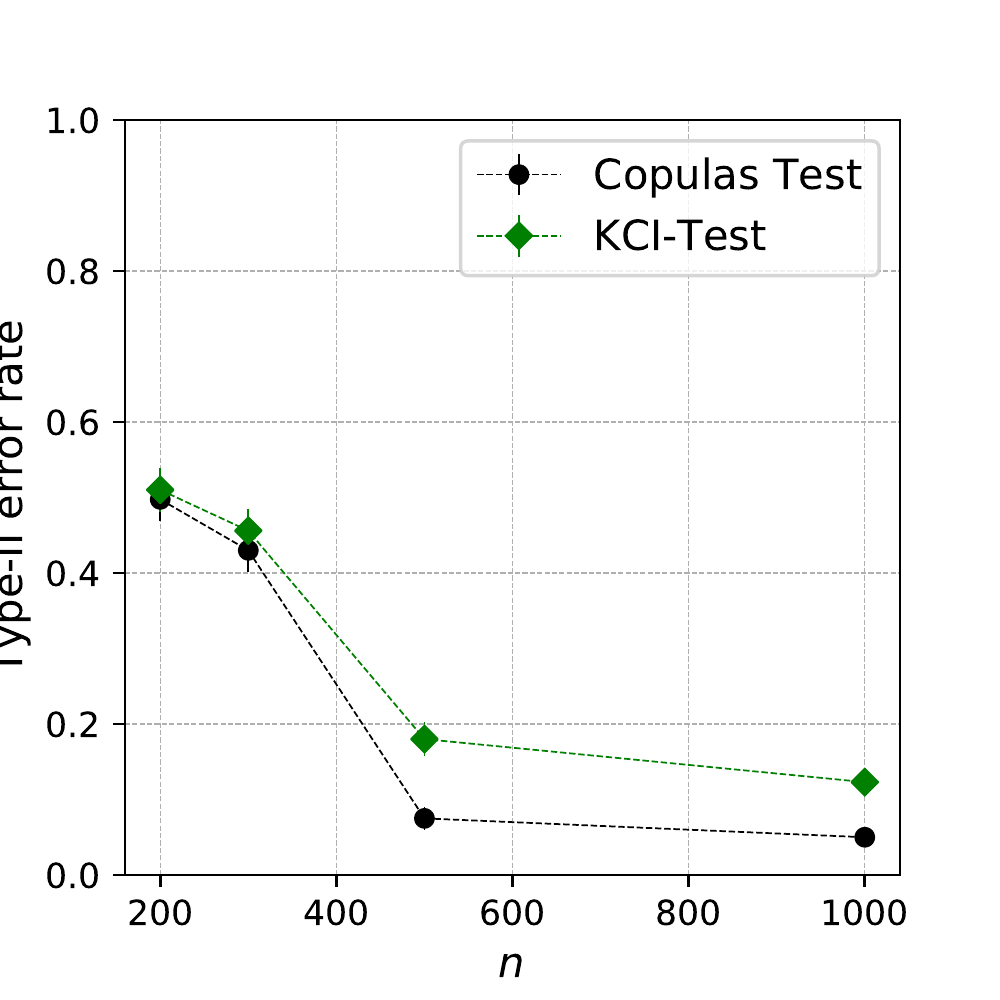}}
\caption{The probability of type-I (\ref{fig:type_II_error_vs_n_pnl}) and type-II (\ref{fig:type_II_error_vs_n_pnl}).}
\label{fig:errors_vs_n_pnl}
\end{figure*}

\subsection{Classification of age groups using functional connectivity}

%In this example, different kinds of functional connectivity between regions of interest are compared. 
%The resulting connectivity coefficients between probabilistic regions of interest (ROIs) are used to discriminate children from adults.
%Connectivity matrices are used as features to distinguish children from adults.
%In general, the tangent space embedding outperforms the standard correlation

In this paragraph, we apply our test in a practical setting, using the
movie watching based brain development dataset
\cite{richardson2018development} obtained from the OpenNeuro
database\footnote{Accession number ds000228.}. The dataset consists in
$50$ patients ($10$ adults and $40$ children).  The fMRI data consists
in measuring the brain activity in 39 Region of Interest (ROI).  For
every patient, $168$ measurements are provided for each ROI.  We
denote for $j \in \{1, \dots, 39\}$ by $X_j$ the variable that
represents the $j^{th}$ region signal value.  Given two ROI $j$ and
$j^\prime$, we seek to test the null hypothesis that $X_j$ and
$X_{j^\prime}$ are conditionally independent given $\mathbf{X}
_{\backslash\{j, j^{\prime}\}}$.  The decisions given by our test
allow us to obtain a map of connections between all the ROI, called
\textit{connectome}, given in Figure \ref{fig:connectome}.  In this
figure, a line is drawn between two ROI whenever our test rejects the
null for these two ROI.  Here, due to the high dimension of the
conditional variable, the margins are no longer estimated using a 
Gaussian kernel as in Section~\ref{sec:generic_example}, but using 
a $k$-nearest neighbours approach. For a given $x$,
the mapping $\hat F_{n,j}(y|x)$ is estimated for every $y\in \mathbb R$
as the proportion of samples $i$ amongst the $k$-nearest nearest neighbours of $x$
which satisfy $Y_{ij}\leq y$. The integer $k$ is select by cross-validation.

As a sanity check, our connectome is used as an input feature 
of a classifier (Linear Support
Vector Classifier (SVC)) in order to distinguish children from adults. 
We estimate the classification accuracy of our classifier using $k$-fold.
The obtained accuracy is 97.4$\%$. This result outperforms the standard
correlation method (91.3$\%$)  and is close to the so-called tangent method
(98.9$\%$) which is known to be fitted
for this task \cite{dadi2019benchmarking}.

%\begin{table}
%\caption{Classification accuracy of a Linear Support Vector Classifier (SVC) using different connectomes as features to discriminate a child from an adult.
%}
%\label{tab:error_in_classif}
%\vskip 0.05in
%\begin{center}
%\begin{small}
%\begin{sc}
%\begin{tabular}{lccc}
%\toprule
% & copulas & tangent & correlation \\
%\midrule
%Accuracy & $97.4 \pm 7$ & $98.7 \pm 8$ & $91.3 \pm 10$ \\
%\bottomrule
%\end{tabular}
%\end{sc}
%\end{small}
%\end{center}
%\vskip -0.1in
%\end{table}

\begin{figure}
\centering
\includegraphics[width=.9\columnwidth]{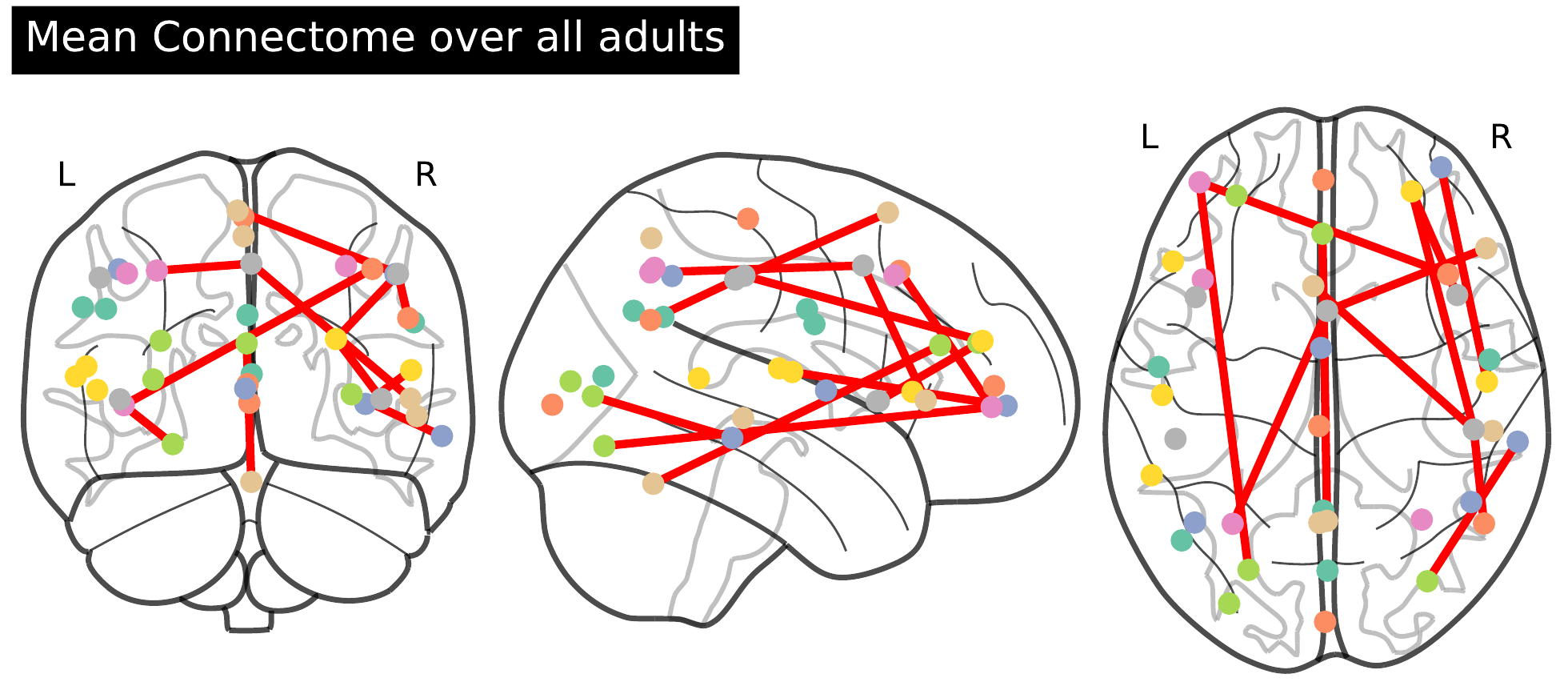}
\caption{Mean connectome provided by our test over all adults.}\label{fig:connectome}
\end{figure}

\section*{Conclusion}
In this work, we have developed test of conditional independence between $Y_1$ and $Y_2$ given $X$ based on weighted partial copulas. First, under general conditions, the consistency of the weighted partial copula process has been established. We have shown that, empirically, our proposed test shows better performance, in terms of power, than recent state-of-the-art conditional independence tests such as the one based on a kernel embedding.

\clearpage
%\section*{Broader Impact}
%The theoretical results presented in this paper do not present any foreseeable societal 
%consequence.

\setlength{\bibsep}{0pt}%{0.3ex plus 0.3ex}
\small
\bibliographystyle{chicago}
\bibliography{b2}

\begin{thebibliography}{}

\bibitem[\protect\citeauthoryear{Bach and Jordan}{Bach and
  Jordan}{2003}]{bach+j:2003}
Bach, F.~R. and M.~I. Jordan (2003).
\newblock Learning graphical models with mercer kernels.
\newblock In {\em Advances in Neural Information Processing Systems}, pp.\
  1033--1040.

\bibitem[\protect\citeauthoryear{Bell, Pattison, and Withers}{Bell
  et~al.}{1988}]{bell1988conditional}
Bell, R.~C., P.~E. Pattison, and G.~P. Withers (1988).
\newblock Conditional independence in a clustered item test.
\newblock {\em Applied Psychological Measurement\/}~{\em 12\/}(1), 15--26.

\bibitem[\protect\citeauthoryear{Beran, Bilodeau, and de~Micheaux}{Beran
  et~al.}{2007}]{beran:2007}
Beran, R., M.~Bilodeau, and P.~L. de~Micheaux (2007).
\newblock Nonparametric tests of independence between random vectors.
\newblock {\em Journal of Multivariate Analysis\/}~{\em 98\/}(9), 1805--1824.

\bibitem[\protect\citeauthoryear{Bergsma}{Bergsma}{2010}]{bergsma2010nonparametric}
Bergsma, W. (2010).
\newblock Nonparametric testing of conditional independence by means of the
  partial copula.
\newblock {\em Available at SSRN 1702981\/}.

\bibitem[\protect\citeauthoryear{Berrett, Wang, Barber, and Samworth}{Berrett
  et~al.}{2019}]{berrett2019conditional}
Berrett, T.~B., Y.~Wang, R.~F. Barber, and R.~J. Samworth (2019).
\newblock The conditional permutation test for independence while controlling
  for confounders.
\newblock {\em Journal of the Royal Statistical Society: Series B (Statistical
  Methodology)\/}.

\bibitem[\protect\citeauthoryear{Biau, C{\'e}rou, and Guyader}{Biau
  et~al.}{2010}]{biau2010rate}
Biau, G., F.~C{\'e}rou, and A.~Guyader (2010).
\newblock On the rate of convergence of the bagged nearest neighbor estimate.
\newblock {\em Journal of Machine Learning Research\/}~{\em 11\/}(2).

\bibitem[\protect\citeauthoryear{Bouezmarni, Rombouts, and Taamouti}{Bouezmarni
  et~al.}{2012}]{bouezmarni:2012}
Bouezmarni, T., J.~V. Rombouts, and A.~Taamouti (2012).
\newblock Nonparametric copula-based test for conditional independence with
  applications to granger causality.
\newblock {\em Journal of Business \& Economic Statistics\/}~{\em 30\/}(2),
  275--287.

\bibitem[\protect\citeauthoryear{B{\"u}cher and Dette}{B{\"u}cher and
  Dette}{2010}]{bucher+d:2010}
B{\"u}cher, A. and H.~Dette (2010).
\newblock A note on bootstrap approximations for the empirical copula process.
\newblock {\em Statistics \& probability letters\/}~{\em 80\/}(23-24),
  1925--1932.

\bibitem[\protect\citeauthoryear{Candes, Fan, Janson, and Lv}{Candes
  et~al.}{2018}]{candes2018panning}
Candes, E., Y.~Fan, L.~Janson, and J.~Lv (2018).
\newblock Panning for gold:‘model-x’knockoffs for high dimensional
  controlled variable selection.
\newblock {\em Journal of the Royal Statistical Society: Series B (Statistical
  Methodology)\/}~{\em 80\/}(3), 551--577.

\bibitem[\protect\citeauthoryear{Chaudhuri and Dasgupta}{Chaudhuri and
  Dasgupta}{2010}]{chaudhuri2010rates}
Chaudhuri, K. and S.~Dasgupta (2010).
\newblock Rates of convergence for the cluster tree.
\newblock In {\em NIPS}, pp.\  343--351. Citeseer.

\bibitem[\protect\citeauthoryear{Dabrowska}{Dabrowska}{1989}]{dabrowska1989uniform}
Dabrowska, D.~M. (1989).
\newblock Uniform consistency of the kernel conditional kaplan-meier estimate.
\newblock {\em The Annals of Statistics\/}, 1157--1167.

\bibitem[\protect\citeauthoryear{Dadi, Rahim, Abraham, Chyzhyk, Milham,
  Thirion, Varoquaux, Initiative, et~al.}{Dadi
  et~al.}{2019}]{dadi2019benchmarking}
Dadi, K., M.~Rahim, A.~Abraham, D.~Chyzhyk, M.~Milham, B.~Thirion,
  G.~Varoquaux, A.~D.~N. Initiative, et~al. (2019).
\newblock Benchmarking functional connectome-based predictive models for
  resting-state fmri.
\newblock {\em Neuroimage\/}~{\em 192}, 115--134.

\bibitem[\protect\citeauthoryear{Dawid}{Dawid}{1979}]{dawid:1979}
Dawid, A.~P. (1979).
\newblock Conditional independence in statistical theory.
\newblock {\em Journal of the Royal Statistical Society: Series B
  (Methodological)\/}~{\em 41\/}(1), 1--15.

\bibitem[\protect\citeauthoryear{Deheuvels}{Deheuvels}{1981}]{deheuvels:1981}
Deheuvels, P. (1981).
\newblock An asymptotic decomposition for multivariate distribution-free tests
  of independence.
\newblock {\em Journal of Multivariate Analysis\/}~{\em 11\/}(1), 102--113.

\bibitem[\protect\citeauthoryear{Derumigny and Fermanian}{Derumigny and
  Fermanian}{2020}]{derumigny2020conditional}
Derumigny, A. and J.-D. Fermanian (2020).
\newblock Conditional empirical copula processes and generalized dependence
  measures.
\newblock {\em arXiv preprint arXiv:2008.09480\/}.

\bibitem[\protect\citeauthoryear{Doran, Muandet, Zhang, and
  Sch{\"o}lkopf}{Doran et~al.}{2014}]{doran2014permutation}
Doran, G., K.~Muandet, K.~Zhang, and B.~Sch{\"o}lkopf (2014).
\newblock A permutation-based kernel conditional independence test.
\newblock In {\em UAI}, pp.\  132--141.

\bibitem[\protect\citeauthoryear{Durante and Sempi}{Durante and
  Sempi}{2015}]{durante2015principles}
Durante, F. and C.~Sempi (2015).
\newblock {\em Principles of copula theory}.
\newblock CRC press.

\bibitem[\protect\citeauthoryear{Einmahl and Mason}{Einmahl and
  Mason}{2000}]{einmahl2000empirical}
Einmahl, U. and D.~M. Mason (2000).
\newblock An empirical process approach to the uniform consistency of
  kernel-type function estimators.
\newblock {\em Journal of Theoretical Probability\/}~{\em 13\/}(1), 1--37.

\bibitem[\protect\citeauthoryear{Fan and Gijbels}{Fan and
  Gijbels}{1996}]{fan1996}
Fan, J. and I.~Gijbels (1996).
\newblock {\em Local polynomial modelling and its applications}, Volume~66 of
  {\em Monographs on Statistics and Applied Probability}.
\newblock Chapman \& Hall, London.

\bibitem[\protect\citeauthoryear{Fermanian, Radulovic, Wegkamp,
  et~al.}{Fermanian et~al.}{2004}]{fermanian+r+w:2004}
Fermanian, J.-D., D.~Radulovic, M.~Wegkamp, et~al. (2004).
\newblock Weak convergence of empirical copula processes.
\newblock {\em Bernoulli\/}~{\em 10\/}(5), 847--860.

\bibitem[\protect\citeauthoryear{Fukumizu, Bach, and Jordan}{Fukumizu
  et~al.}{2004}]{fukumizu2004dimensionality}
Fukumizu, K., F.~R. Bach, and M.~I. Jordan (2004).
\newblock Dimensionality reduction for supervised learning with reproducing
  kernel hilbert spaces.
\newblock {\em Journal of Machine Learning Research\/}~{\em 5\/}(Jan), 73--99.

\bibitem[\protect\citeauthoryear{Fukumizu, Gretton, Sun, and
  Sch{\"o}lkopf}{Fukumizu et~al.}{2007}]{fukumizu2007kernel}
Fukumizu, K., A.~Gretton, X.~Sun, and B.~Sch{\"o}lkopf (2007).
\newblock Kernel measures of conditional dependence.
\newblock In {\em NIPS}, Volume~20, pp.\  489--496.

\bibitem[\protect\citeauthoryear{Genest, Quessy, and R{\'e}millard}{Genest
  et~al.}{2006}]{genest:2006}
Genest, C., J.-F. Quessy, and B.~R{\'e}millard (2006).
\newblock Local efficiency of a cram{\'e}r--von mises test of independence.
\newblock {\em Journal of Multivariate Analysis\/}~{\em 97\/}(1), 274--294.

\bibitem[\protect\citeauthoryear{Genest and R{\'e}millard}{Genest and
  R{\'e}millard}{2004}]{genest+r:2004}
Genest, C. and B.~R{\'e}millard (2004).
\newblock Test of independence and randomness based on the empirical copula
  process.
\newblock {\em Test\/}~{\em 13\/}(2), 335--369.

\bibitem[\protect\citeauthoryear{Gijbels, Omelka, and Veraverbeke}{Gijbels
  et~al.}{2015}]{omelka+g+v:15}
Gijbels, I., M.~Omelka, and N.~Veraverbeke (2015).
\newblock Estimation of a copula when a covariate affects only marginal
  distributions.
\newblock {\em Scandinavian Journal of Statistics\/}~{\em 42\/}(4), 1109--1126.

\bibitem[\protect\citeauthoryear{Gijbels, Veraverbeke, and Omelka}{Gijbels
  et~al.}{2011}]{gijbels+v+o:2011:csda}
Gijbels, I., N.~Veraverbeke, and M.~Omelka (2011).
\newblock Conditional copulas, association measures and their applications.
\newblock {\em Computational Statistics \& Data Analysis\/}~{\em 55\/}(5),
  1919--1932.

\bibitem[\protect\citeauthoryear{Gin{\'e} and Guillou}{Gin{\'e} and
  Guillou}{2002}]{gine+g:02}
Gin{\'e}, E. and A.~Guillou (2002).
\newblock Rates of strong uniform consistency for multivariate kernel density
  estimators.
\newblock {\em Ann. Inst. H. Poincar\'e Probab. Statist.\/}~{\em 38\/}(6),
  907--921.
\newblock En l'honneur de J. Bretagnolle, D. Dacunha-Castelle, I. Ibragimov.

\bibitem[\protect\citeauthoryear{Giné and Guillou}{Giné and
  Guillou}{2001}]{gineConsistencyKernelDensity2001}
Giné, E. and A.~Guillou (2001).
\newblock On consistency of kernel density estimators for randomly censored
  data: Rates holding uniformly over adaptive intervals.
\newblock ~{\em 37\/}(4), 503--522.

\bibitem[\protect\citeauthoryear{Gretton, Fukumizu, Teo, Song, Sch{\"o}lkopf,
  and Smola}{Gretton et~al.}{2008}]{gretton2008kernel}
Gretton, A., K.~Fukumizu, C.~H. Teo, L.~Song, B.~Sch{\"o}lkopf, and A.~J. Smola
  (2008).
\newblock A kernel statistical test of independence.
\newblock In {\em Advances in neural information processing systems}, pp.\
  585--592.

\bibitem[\protect\citeauthoryear{Gy{\"o}rfi, Kohler, Krzyzak, and
  Walk}{Gy{\"o}rfi et~al.}{2006}]{gyorfi2006distribution}
Gy{\"o}rfi, L., M.~Kohler, A.~Krzyzak, and H.~Walk (2006).
\newblock {\em A distribution-free theory of nonparametric regression}.
\newblock Springer Science \& Business Media.

\bibitem[\protect\citeauthoryear{Hall and Wilson}{Hall and
  Wilson}{1991}]{hall+w:1991}
Hall, P. and S.~R. Wilson (1991).
\newblock Two guidelines for bootstrap hypothesis testing.
\newblock {\em Biometrics\/}, 757--762.

\bibitem[\protect\citeauthoryear{Hansen}{Hansen}{2008}]{hansen2008uniform}
Hansen, B.~E. (2008).
\newblock Uniform convergence rates for kernel estimation with dependent data.
\newblock {\em Econometric Theory\/}, 726--748.

\bibitem[\protect\citeauthoryear{Hoeffding}{Hoeffding}{1948}]{hoeffding:1948}
Hoeffding, W. (1948).
\newblock A non-parametric test of independence.
\newblock {\em The annals of mathematical statistics\/}, 546--557.

\bibitem[\protect\citeauthoryear{Huber and Melly}{Huber and
  Melly}{2015}]{huber2015test}
Huber, M. and B.~Melly (2015).
\newblock A test of the conditional independence assumption in sample selection
  models.
\newblock {\em Journal of Applied Econometrics\/}~{\em 30\/}(7), 1144--1168.

\bibitem[\protect\citeauthoryear{Jiang}{Jiang}{2019}]{jiang2019non}
Jiang, H. (2019).
\newblock Non-asymptotic uniform rates of consistency for k-nn regression.
\newblock In {\em Proceedings of the AAAI Conference on Artificial
  Intelligence}, Volume~33, pp.\  3999--4006.

\bibitem[\protect\citeauthoryear{Kendall}{Kendall}{1948}]{kendall:1948}
Kendall, M.~G. (1948).
\newblock Rank correlation methods.

\bibitem[\protect\citeauthoryear{Kojadinovic and Holmes}{Kojadinovic and
  Holmes}{2009}]{kojadinovic:2009}
Kojadinovic, I. and M.~Holmes (2009).
\newblock Tests of independence among continuous random vectors based on
  cram{\'e}r--von mises functionals of the empirical copula process.
\newblock {\em Journal of Multivariate Analysis\/}~{\em 100\/}(6), 1137--1154.

\bibitem[\protect\citeauthoryear{Koller and Friedman}{Koller and
  Friedman}{2009}]{koller2009probabilistic}
Koller, D. and N.~Friedman (2009).
\newblock {\em Probabilistic graphical models: principles and techniques}.
\newblock MIT press.

\bibitem[\protect\citeauthoryear{Lavergne and Patilea}{Lavergne and
  Patilea}{2013}]{lavergne:2013}
Lavergne, P. and V.~Patilea (2013).
\newblock Smooth minimum distance estimation and testing with conditional
  estimating equations: uniform in bandwidth theory.
\newblock {\em Journal of Econometrics\/}~{\em 177\/}(1), 47--59.

\bibitem[\protect\citeauthoryear{Lee, Li, and Zhao}{Lee
  et~al.}{2016}]{lee+l+z:2016}
Lee, K.-Y., B.~Li, and H.~Zhao (2016).
\newblock Variable selection via additive conditional independence.
\newblock {\em Journal of the Royal Statistical Society: Series B (Statistical
  Methodology)\/}~{\em 78\/}(5), 1037--1055.

\bibitem[\protect\citeauthoryear{Li}{Li}{2018}]{li:2018}
Li, B. (2018).
\newblock {\em Sufficient dimension reduction: Methods and applications with
  R}.
\newblock Chapman and Hall/CRC.

\bibitem[\protect\citeauthoryear{Li and Fan}{Li and
  Fan}{2020}]{li2020nonparametric}
Li, C. and X.~Fan (2020).
\newblock On nonparametric conditional independence tests for continuous
  variables.
\newblock {\em Wiley Interdisciplinary Reviews: Computational
  Statistics\/}~{\em 12\/}(3), e1489.

\bibitem[\protect\citeauthoryear{Major}{Major}{2006}]{major2006estimate}
Major, P. (2006).
\newblock An estimate on the supremum of a nice class of stochastic integrals
  and u-statistics.
\newblock {\em Probability Theory and Related Fields\/}~{\em 134\/}(3),
  489--537.

\bibitem[\protect\citeauthoryear{Markowetz and Spang}{Markowetz and
  Spang}{2007}]{markowetz2007inferring}
Markowetz, F. and R.~Spang (2007).
\newblock Inferring cellular networks--a review.
\newblock {\em BMC bioinformatics\/}~{\em 8\/}(6), S5.

\bibitem[\protect\citeauthoryear{Nolan and Pollard}{Nolan and
  Pollard}{1987}]{nolan+p:1987}
Nolan, D. and D.~Pollard (1987).
\newblock {$U$}-processes: rates of convergence.
\newblock {\em The Annals of Statistics\/}~{\em 15\/}(2), 780--799.

\bibitem[\protect\citeauthoryear{Portier and Segers}{Portier and
  Segers}{2018}]{portier+s:2018}
Portier, F. and J.~Segers (2018).
\newblock On the weak convergence of the empirical conditional copula under a
  simplifying assumption.
\newblock {\em Journal of Multivariate Analysis\/}~{\em 166}, 160--181.

\bibitem[\protect\citeauthoryear{R{\'e}millard and Scaillet}{R{\'e}millard and
  Scaillet}{2009}]{remillard:2009}
R{\'e}millard, B. and O.~Scaillet (2009).
\newblock Testing for equality between two copulas.
\newblock {\em Journal of Multivariate Analysis\/}~{\em 100\/}(3), 377--386.

\bibitem[\protect\citeauthoryear{Richardson, Lisandrelli, Riobueno-Naylor, and
  Saxe}{Richardson et~al.}{2018}]{richardson2018development}
Richardson, H., G.~Lisandrelli, A.~Riobueno-Naylor, and R.~Saxe (2018).
\newblock Development of the social brain from age three to twelve years.
\newblock {\em Nature communications\/}~{\em 9\/}(1), 1--12.

\bibitem[\protect\citeauthoryear{Runge}{Runge}{2017}]{runge2017conditional}
Runge, J. (2017).
\newblock Conditional independence testing based on a nearest-neighbor
  estimator of conditional mutual information.
\newblock {\em arXiv preprint arXiv:1709.01447\/}.

\bibitem[\protect\citeauthoryear{Ruschendorf}{Ruschendorf}{1976}]{ruschendorf:1976}
Ruschendorf, L. (1976).
\newblock Asymptotic distributions of multivariate rank order statistics.
\newblock {\em The Annals of Statistics\/}, 912--923.

\bibitem[\protect\citeauthoryear{Ruymgaart}{Ruymgaart}{1974}]{ruymgaart:1974}
Ruymgaart, F.~H. (1974).
\newblock Asymptotic normality of nonparametric tests for independence.
\newblock {\em The Annals of Statistics\/}, 892--910.

\bibitem[\protect\citeauthoryear{Ruymgaart and van Zuijlen}{Ruymgaart and van
  Zuijlen}{1978}]{ruymgaart:1978}
Ruymgaart, F.~H. and M.~van Zuijlen (1978).
\newblock Asymptotic normality of multivariate linear rank statistics in the
  non-iid case.
\newblock {\em The Annals of Statistics\/}, 588--602.

\bibitem[\protect\citeauthoryear{Segers}{Segers}{2012}]{segers:2012}
Segers, J. (2012).
\newblock Asymptotics of empirical copula processes under non-restrictive
  smoothness assumptions.
\newblock {\em Bernoulli\/}~{\em 18\/}(3), 764--782.

\bibitem[\protect\citeauthoryear{Sen, Suresh, Shanmugam, Dimakis, and
  Shakkottai}{Sen et~al.}{2017}]{sen2017model}
Sen, R., A.~T. Suresh, K.~Shanmugam, A.~G. Dimakis, and S.~Shakkottai (2017).
\newblock Model-powered conditional independence test.
\newblock In {\em Advances in Neural Information Processing Systems}, pp.\
  2951--2961.

\bibitem[\protect\citeauthoryear{Stone}{Stone}{1977}]{stone1977consistent}
Stone, C.~J. (1977).
\newblock Consistent nonparametric regression.
\newblock {\em The annals of statistics\/}, 595--620.

\bibitem[\protect\citeauthoryear{Stone}{Stone}{1982}]{stone1982optimal}
Stone, C.~J. (1982).
\newblock Optimal global rates of convergence for nonparametric regression.
\newblock {\em The annals of statistics\/}, 1040--1053.

\bibitem[\protect\citeauthoryear{Su and White}{Su and White}{2007}]{su:2007}
Su, L. and H.~White (2007).
\newblock A consistent characteristic function-based test for conditional
  independence.
\newblock {\em Journal of Econometrics\/}~{\em 141\/}(2), 807--834.

\bibitem[\protect\citeauthoryear{Su and White}{Su and
  White}{2008}]{su2008nonparametric}
Su, L. and H.~White (2008).
\newblock A nonparametric hellinger metric test for conditional independence.
\newblock {\em Econometric Theory\/}~{\em 24\/}(4), 829--864.

\bibitem[\protect\citeauthoryear{Talagrand}{Talagrand}{1996}]{talagrandNewConcentrationInequalities1996}
Talagrand, M. (1996).
\newblock New concentration inequalities in product spaces.
\newblock {\em Inventiones mathematicae\/}~{\em 126\/}(3), 505--563.

\bibitem[\protect\citeauthoryear{van~der Vaart and Wellner}{van~der Vaart and
  Wellner}{1996}]{wellner1996}
van~der Vaart, A.~W. and J.~A. Wellner (1996).
\newblock {\em Weak Convergence and Empirical Processes. With Applications to
  Statistics}.
\newblock Springer Series in Statistics. New York: Springer-Verlag.

\bibitem[\protect\citeauthoryear{Veraverbeke, Omelka, and Gijbels}{Veraverbeke
  et~al.}{2011}]{veraverbeke+o+g:2011}
Veraverbeke, N., M.~Omelka, and I.~Gijbels (2011).
\newblock Estimation of a conditional copula and association measures.
\newblock {\em Scandinavian Journal of Statistics\/}~{\em 38}, 766--780.

\bibitem[\protect\citeauthoryear{Wenocur and Dudley}{Wenocur and
  Dudley}{1981}]{wenocur1981some}
Wenocur, R.~S. and R.~M. Dudley (1981).
\newblock Some special vapnik-chervonenkis classes.
\newblock {\em Discrete Mathematics\/}~{\em 33\/}(3), 313--318.

\bibitem[\protect\citeauthoryear{Zhang, Peters, Janzing, and
  Sch\"{o}lkopf}{Zhang et~al.}{2011}]{zhang2012kernel}
Zhang, K., J.~Peters, D.~Janzing, and B.~Sch\"{o}lkopf (2011).
\newblock Kernel-based conditional independence test and application in causal
  discovery.
\newblock In {\em Proceedings of the Twenty-Seventh Conference on Uncertainty
  in Artificial Intelligence}, pp.\  804–813. AUAI Press.

\end{thebibliography}
%\end{document}

\clearpage
\appendix
\begin{appendices}

\section{Proofs of the basic lemmas of Section \ref{sec:def}}
\label{app:proof:basic_lemma}

\subsection{Proof of Lemma \ref{lemma:cara}}\label{app:proof:cara}

The ``only if'' part is obvious. Suppose that the function $W   = 0$ and let $\bu \in [0,1]^2$. Define $g(x) = C( \bu \mid x)  - u_1u_2 )f_X(x)$. We have  $(g\star w) = 0$, a.e. on $\mathbb R^d$, where $\star $ stands for the standard convolution product with respect to the Lebesgue measure. Applying the Fourier transform gives that $\mathcal F(g) \mathcal F (w) = 0$ which, by assumption, yields $\mathcal F(g) = 0$. By the Fourier inversion theorem we obtain that $g = 0$ a.e. on $\mathbb R^d$. That is for any $\bu\in [0,1]^2$ and any $x\in S_X$, $C( \bu \mid x)  = u_1u_2 $. 

\qed

\subsection{Proof of Lemma \ref{eq:statistic_T}}
\label{app:proof:statistic_computation}
Write
\begin{align*}
\hat T_{n} & =n^{-2} \sum_{1\leq i,j \leq n } \int_{[0,1]^2} \xi_{i}(\bu) \xi_j(\bu) \,\diff \bu  \int_{\R^d}   \omega(t -X_i)\omega(t - X_j)\,\diff t,
\end{align*}
where $\xi_{k}(\bu) = \ind _{\{\hat G _ {k1} \,  < \, u_1 \}}\ind _{\{ \hat G_{k2} \, < \,  u_2\}} -  u_1u_2 $. It remains to compute the function $M$. Using the notation $\hat{\bm {G}}_{i}$, we have
\begin{align*}
\int_{[0,1]^2} \xi_{i}(\bu) \xi_j(\bu) \,\diff \bu =& \int_{[0,1]^2}  \ind _{\{ \hat G _ {i1} \,  < \, u_1 \}}\ind _{\{ \hat G_{i2} \, < \, u_2\}} \ind _{\{ \hat G _ {j1} \,  < \, u_1 \}}\ind _{\{ \hat G_{j2} \, < \,  u_2\}} \diff \bu \\ 
&- \int_{[0,1]^2}  \ind _{\{ \hat G _ {i1} \,  < \, u_1 \}}\ind _{\{ \hat G_{i2}  \, < \,   u_2\}} u_1 u_2 \diff \bu \\
&- \int_{[0,1]^2}  \ind _{\{ \hat G _ {j1} \,  < \, u_1 \}}\ind _{\{ \hat G_{j2}  \, < \,  u_2\}} u_1 u_2 \diff \bu + \int_{[0,1]^2} (u_1 u_2)^2 \diff \bu.
\end{align*}
First, let compute the first term of the right hand side. Let notice that the value of the integrand is 1 if $u_1 > \hat G _ {i1}\vee \hat G _ {j1}$ and $u_2 > \hat G _ {u2} \vee \hat G _ {j2}$ and $0$ otherwise. Thus we obtain for this term:
\begin{align}
\label{eq:1}
\int_{[0,1]^2}&  \ind_{ u_1 > \hat G _ {i1}\vee \hat G _ {j1}} \ind_{ u_2 > \hat G _ {i2}\vee \hat G _ {j2}} \diff \bu = \left( 1 - \hat G _ {i1}\vee \hat G _ {j1} \right)\left( 1 - \hat G _ {i2}\vee \hat G _ {j2}\right).
\end{align}
Now let derive the second integral term of the right hand side, the third term will follow directly. 

\begin{align}
\label{eq:2}
\int_{[0,1]^2}  \ind _{\{ \hat G _ {i1} \,  < \,  u_1 \}}\ind _{\{ \hat G_{i2} \, < \,  u_2\}} u_1 u_2 \diff \bu &= 
\int_{[0,1]}  \ind _{\{ \hat G _ {i1} \,  < \, u_1 \}}u_1\diff u_1 \int_{[0,1]} \ind _{\{ \hat G_{i2}  \, < \,   u_2\}} u_2 \diff u_2 \nonumber \\ 
&= \frac{1}{4} \left(1 - \hat G _ {i1}^2 \right) \left(1 - \hat G _ {i2}^2 \right).
\end{align}
By combining \eqref{eq:1} and \eqref{eq:2} we obtain the desired result.
\qed

\section{Proof of Theorem \ref{theorem:consistency}}
\label{app:proof:consistency}

The proof is inspired from the proof Theorem 1 in  \cite{omelka+g+v:15}.
 
\paragraph{Notation.}

We use  notation from empirical process theory. Let $P_n=n^{-1}\sum_{i=1}^n\delta_{(X_i,\bY_i)}$ denote the empirical measure. For a  function $f$ and a probability measure $Q$, write $Qf=\int f\, \diff Q$. The empirical process is
\begin{align*}
\mathbb G_n =n^{1/2}  (P_n -P).
\end{align*}
For any pair of cumulative distribution  functions $F_1$ and $F_2$  on $\R$, put $\bm{F} (\by)= (F_{1}(y_1)  , F_{2}(y_2)) $ for  $\by = (y_1,y_2)\in \R^2$ and $\bm{F}^- (\bu)= (F_{1}^-(u_1)  , F_{2}^-(u_2)) $ for  $\bu=(u_1,u_2)\in [0,1]^2$.

\paragraph{Preliminary results.}

\begin{fact}\label{prel:inv}
\begin{align*}
\sup_{u_j\in [0,1]  }   | \hat  G_{n,j}  ( u_j )    - u_j  | = O_{\mathbb P} (r_n + n^{-1/2} ).
\end{align*}
%This results from (G\ref{cond:smoothnessdensity1})  and (H\ref{cond:high_level_consistency}) as shown in page 170 of \cite{portier+s:2018}.
\end{fact}
\begin{proof}
With probability at least $1- \epsilon/2$,
\begin{align*}
\sup_{u_j\in [0,1],\, x\in S_X }   | \hat  F_{n,j}  ( u_j |x )    - F ( u_j|x)  | \leq M r_n.
\end{align*}
Moreover, from Donsker's theorem, we have with probability $1-\epsilon/2$ 
\begin{align*}
\sup_{u_j\in \mathbb R } (P_n -P) \{\ind_{\{   {F}_{j}\, \leq\,  u_j   \}} \}  \leq M n^{-1/2}.
\end{align*}
As a consequence, we have with probability $1-\epsilon$,
\begin{align*}
 \hat  G_{n,j}  ( u_j ) &=  P_n  \ind_{\{  \hat {{F}}_{n,j}\, \leq\,  u_j  \}}\\
 &\leq  P_n  \ind_{\{   {F}_{j}\, \leq\,  u_j + M r_n  \}} \\
 & \leq P \ind_{\{   {F}_{j}\, \leq\,  u_j + M r_n  \}}  + Mn^{-1/2}.
\end{align*}
Using that $  P \ind_{\{   {F}_{j}\, \leq\,  u_j + M r_n  \}} = u_j + M r_n$, we have shown that
\begin{align*}
\hat  G_{n,j}  ( u_j )  - u_j \leq M ( r_n + n^{-1/2} ) 
\end{align*}
The lower bound $\hat  G_{n,j}  ( u_j )  - u_j \geq  - M ( r_n + n^{-1/2} ) $ can be obtained in the same way.
\end{proof}

\begin{fact}\label{prel:ranks}
\begin{align*}
\sup_{u_j\in [0,1],\, x\in S_X }   | \hat  G_{n,j}  ( \hat {{F}}_{n,j} (u_j |x  ) )     -  F ( u_j|x)  | = O_{\mathbb P} (r_n + n^{-1/2} ) .
\end{align*}
%This results from (G\ref{cond:smoothnessdensity1})  and (H\ref{cond:high_level_consistency}) as shown in page 170 of \cite{portier+s:2018}.
\end{fact}
\begin{proof}
Using the triangle inequality and Fact \ref{prel:inv} leads to the result.
\end{proof}

\begin{fact}\label{prel:donsker}
\begin{align*}
\sup_{\bu \in [0, 1]^2 , \, t\in\R  }   | (P_n - P)   \left\{w_t  \ind_{\{   \bm{F} \, \leq\,  \bu  \}}  \right\}         | = O_{\mathbb P} ( n^{-1/2} ).
\end{align*}
%This results from (G\ref{cond:smoothnessdensity1})  and (H\ref{cond:high_level_consistency}) as shown in page 170 of \cite{portier+s:2018}.
\end{fact}

\begin{proof}
 Since $\tilde w $ is of bounded variation, the function class $\{x\mapsto w_t(x) \, :\, t\in \R\}$ is Euclidean (or VC) with constant envelop $C_w = \sup_{x\in \mathbb R} |w(x) |$ \citep[Lemma 22, (i)]{nolan+p:1987}, i.e., the covering numbers are polynomials. Moreover, the class of indicator functions is also Euclidean \citep[Example 2.5.4]{wellner1996}. This implies that both classes have finite entropy integrals and therefore are Donsker \citep[Chapter 2.1, equation (2.1.7)]{wellner1996}. Using the preservation of the Donsker property through products and sums  \citep[Example~2.10.7 and 2.10.8]{wellner1996}, the class $\{ (\by,x)  \mapsto  w_t(x) \mathds{1} _{\{   \bm F  (\by| x)  \leq \bu  \}} \,:\, t\in \mathbb R,\, \bu \in [0,1]^2 \} $ is Donsker. As a result, the process $\{ \mathbb G_n   \{w_t  \ind_{\{   \bm{F} \, \leq\,  \bu  \}}  \}_{\bu \in [0, 1]^2 , \, t\in\R }$ converges weakly to a tight Gaussian process in $\ell^\infty( [0, 1] ^2 \times \R) $.
\end{proof}

%\paragraph{Decomposition.}
%
%We shall rely on the following decomposition
%\begin{align*}
%W_n(\bu,t)  - W(\bu,t)  &= P_n  \left\{w_t  (\ind_{\{  \hat {\bm{F}}_{n}\, \leq\,  \, \hat{\bm{G}}^-_{n} (\bu)  \}} -  \ind_{\{   \bm{F} \, \leq\,  \bu  \}}  ) \right\}  + P_n    \left\{w_t  \ind_{\{   \bm{F} \, \leq\,  \bu  \}}  \right\} -W(\bu,t)      
% \end{align*}

\paragraph{End of the proof.}

We need to show that $W_n(\bu, t) - W(\bu,t)  = O_{\mathbb P} (r_n + n^{-1/2} )$. Let $\epsilon >0$. From Fact \ref{prel:ranks} and Fact \ref{prel:donsker}, we have with probability $1-\epsilon$, %\ind_{\{  {\bm{F}}\, \leq\,  \,  \bu  \}} 
\begin{align*}
W_n(\bu, t)  &=P_n  \left\{w_t  (\ind_{\{  \hat {\bm{F}}_{n}\, \leq\,   \hat{\bm{G}}^-_{n} (\bu)  \}}    \right\} \\
&= P_n  \left\{w_t  (\ind_{\{  \hat{\bm{G}} _{n}  \hat {\bm{F}}_{n}\, < \, \bu + n^{-1}   \}} \right\}\\
& \leq  P_n  \left\{w_t  \ind_{\{  {\bm{F}}\, \leq  \, \bu + a_n \}} \right\} \\
& = P \left\{w_t  (\ind_{\{  {\bm{F}}\, \leq  \, \bu + a_n ) \}} \right\} + Mn^{-1/2}\\
&= W(\bu,t) +L_n+  Mn^{-1/2}
\end{align*}
with
\begin{align*}
&L_n =  P \left\{w_t  (\ind_{\{  {\bm{F}}\, \leq  \, \bu + a_n \}} -   \ind_{\{  {\bm{F}}\, \leq  \, \bu  \}} )\right\}\\
&a_n = n^{-1}  + M (r_n + n^{-1/2} )
\end{align*}
for some $M>0$.
Now, because of the $1$-Lipschitz properties of copulas \citep[Theorem 1.5.1]{durante2015principles}, we have
\begin{align*}
L_n & =    \int w_t(x) \left\{  C (   \bu + a_n \mid x ) - C (   \bu  \mid x )  \right\}  \, f_X (x) \,  \diff x  \leq  2C_w a_n  .
\end{align*}
Consequently, we have shown that with probability $1-\epsilon$,
\begin{align*}
P_n  \left\{w_t  \ind_{\{  \hat {\bm{F}}_{n}\, \leq\,   \hat{\bm{G}}^-_{n} (\bu)  \}}  \right\} \leq   W(\bu,t) +   2C_w a_n +  Mn^{-1/2} .
\end{align*}
Proceeding the same way, we obtain $P_n  \left\{w_t  \ind_{\{  \hat {\bm{F}}_{n}\, \leq\,   \hat{\bm{G}}^-_{n} (\bu)  \}}  \right\} \geq   W(\bu,t) -   2C_w a_n - Mn^{-1/2} $ which implies the desired the result.

\section{Proof of Theorem \ref{theorem:consistency2}}
 
 In virtue of Lemma \ref{lemma:cara}, it suffices to show that 
 \begin{align*}
 \int (W_n(\bu, t) - W(\bu, t)) ^2 \diff t = O_{\mathbb P} ( r_n + n^{-1} ) .
 \end{align*} 
We start recalling that  $ \tilde W_n (  \bu - a_n ,t ) \leq W_n(\bu, t)  \leq \tilde W_n (  \bu + a_n ,t )  $ with 
\begin{align*}
\tilde W_n (  \bu  ,t ) = P_n  \left\{w_t  \ind_{\{  {\bm{F}}\, \leq  \, \bu \}} \right\},
\end{align*}
as established in the proof of Theorem \ref{theorem:consistency}. As a consequence
\begin{align*}
&\int (W_n(\bu, t) - W(\bu, t)) ^2 \diff t \\
&\leq  \max_{z\in [ -1,1]}  \int ( \tilde W_n(\bu + za_n , t) -W(\bu, t)) ^2 \diff t \\
&\leq 2  \max_{z\in [ -1,1]} \left(  \int ( \tilde W_n(\bu + za_n , t) - W(\bu + za_n  , t) )^2 \diff t +  \int (  W (\bu + za_n , t)  -W(\bu, t)) ^2 \diff t \right) \\
& \leq 2\sup_{\bu\in \mathbb R^2  }  \int ( \tilde W_n(\bu , t) -\tilde W_n(\bu  , t) )^2 \diff t  + 2\max_{z\in [ -1,1]} \int (  W (\bu + za_n , t)  -W(\bu, t)) ^2 \diff t 
\end{align*}
Each terms are handled separately. One has that
\begin{align*}
&   \int ( \tilde W_n(\bu , t) -\tilde W_n(\bu  , t) )^2 \diff t \\
 &= 2n^{-2}  \sum_{i<j}  \int (f_i(\bu, t)  - E f_i(\bu, t)  )(f_j(\bu, t)  - E f_j(\bu, t)  ) \diff t + n^{-2} \sum_{i=1}^n   \int (f_i(\bu, t)  - E f_i(\bu, t)  ) ^2  \diff t \\
& \leq n^{-2} \sum_{i<j}  \int (f_i(\bu, t)  - E f_i(\bu, t)  )(f_j(\bu, t)  - E f_j(\bu, t)  ) \diff t +  2 n^{-1}  f_\infty
\end{align*}
with $f_i(\bu, t)  = w_t (X_i)  \ind_{\{  {\bm{F}}(\bY_i | X_i ) \, \leq  \, \bu \}} $ and $ f_\infty = \sup_{x\in S_X} \int w_t(x)^2  \diff t\leq  C_w \|w\|_1 $. We now recognize a degenerate $U$-statistics in the above right-hand side term. Indeed, using Fubini's theorem,
\begin{align*}
&\int (f_i(\bu, t)  - E f_i(\bu, t)  )(f_j(\bu, t)  - E f_j(\bu, t)  ) \, \diff t \\
&=   \int  ( g_{i,j} (\bu, t)  - E [ g_{i,j} (\bu, t)  | i] - E [ g_{i,j} (\bu, t)  | j] + E [ g_{i,j} (\bu, t)  ] ) \, \diff t\\
& = h_{i,j} (\bu)  - E [ h_{i,j} (\bu)  | i] - E [ h_{i,j} (\bu)  | j] + E [ h_{i,j} (\bu)  ]
\end{align*}
with 
\begin{align*}
&g_{i,j} (\bu, t)  =   w_t (X_i) w_t(X_j) \ind_{\{  {\bm{F}}(\bY_i | X_i ) \, \leq  \, \bu \}} \ind_{\{  {\bm{F}}(\bY_j | X_j ) \, \leq  \, \bu \}} \\
&h_{i,j} (\bu)  =   w^\star (X_i - X_j )  \ind_{\{  {\bm{F}}(\bY_i | X_i ) \, \leq  \, \bu \}} .
\end{align*}
As the function class $\{h_{i,j} \,:\, \bu \in \mathbb R^2\}$ is of VC type (see Corollary 21 in \cite{nolan+p:1987}), 
we are in position to apply Theorem 2 in \cite{major2006estimate}, leading to
\begin{align*}
\sup_{\bu\in \mathbb R^2  }   \left| \sum_{i<j}  \int (f_i(\bu, t)  - E f_i(\bu, t)  )(f_j(\bu, t)  - E f_j(\bu, t)  ) \diff t \right|  = O_{\mathbb P} (n) .
\end{align*}
As a consequence $ \sup_{\bu\in \mathbb R^2  }   \int ( \tilde W_n(\bu , t) -\tilde W_n(\bu  , t) )^2 \diff t =  O_{\mathbb P} (n^{-1})$.  Now we bound the remaining term, $\max_{z\in [ -1,1]} \int (  W (\bu + za_n , t)  -W(\bu, t)) ^2 \diff t$. As in the proof of Theorem  \ref{theorem:consistency}, we use the Lipschitz property of copulas. It gives
\begin{align*}
|W (\bu + za_n , t)  -W(\bu, t)| \leq  2 a_n  \int w_t(x) f_X(x) \, \diff x.
\end{align*}
Then applying Jensen's inequality gives that
\begin{align*}
\int | W (\bu + za_n , t)  -W(\bu, t) |^2 \, \diff t &\leq 4 a_n^2   \int \int w_t(x)^2 \, \diff t  f_X(x) \, \diff x\leq 4 C_w \|w\|_1 a_n^2  .
\end{align*}

\section{Proof of Theorem \ref{consistenc:kNN}}

\paragraph{Auxiliary results.}

%The following inequality follows from the well-known Chernoff bound \citep{boucheron2013concentration} (see \citep{ausset2020nearest} for a detailed proof of the next statement).
%%
%%\begin{lemma}\label{lemma=chernoff}
%%Let $(Z_i)_{i\geq 1}$ be a sequence of i.i.d. random variables valued in $\{0,1\}$. Set $\mu =  n \mathbb E [Z_1]$ and $S = \sum_{i=1} ^n Z_i $.   For any $\delta \in (0,1)$ and all $n\geq 1$, we have with probability $1-\delta$:
%%\begin{align*}
%%S \geq \left(1- \sqrt{ \frac{2 \log(1/\delta)  }{  \mu} } \right) \mu  .
%%\end{align*}
%% In addition, for any $\delta \in (0,1)$ and $n\geq 1$, we have with probability $1-\delta$:
%%\begin{align*}
%%S \leq \left(1 +  \sqrt{ \frac{3 \log(1/\delta)   }{  \mu} }  \right) \mu  .
%%\end{align*}
%%\end{lemma}

The following result is Theorem 15 in \cite{chaudhuri2010rates} when applied to the collection of balls 
\begin{align*}
\mathcal  B = \{ B( x, \tau) \,:\, x\in [0,1]^d \, \tau >0\} ,
\end{align*}
where $B(x,\tau ) $ is  the ball with center $x$ and radius $\tau$.

\begin{theorem}\label{lemma=prelim}
 For any $\delta > 0$ and $n\geq 1$, with probability at least $1 -\delta $:
\begin{align*}
-  \beta_n^2  + \beta_n \sqrt {P (B) }     \leq P (B) - P _n (B) \leq  \beta_n \sqrt {P( B) } ,\qquad \forall B \in \mathcal B,
\end{align*}
with $\beta_n  = \sqrt {   c_1 (d/n)  \log(c_2n/\delta )    }$, $c_1>0$ and $c_2>0$ universal constants.
\end{theorem}

Let $(S,\mathcal S)$ be a measurable space. Given a probability measure $Q$ on $\mathcal S$, define the metric space $L_2(Q)$ of $Q$-square-integrable functions, i.e.,
\begin{align*}
L_2(Q) = \{ g: S\mapsto \mathbb R \,:\, \| g\|_{L_2(Q)}   <  \infty\} ,
\end{align*}
where $ \| g \|_{L_2(Q)} ^2 = \int g^2 dQ$.  Given $\mathcal G \subset L_2(Q)$, the $\epsilon$-covering number $\mathcal N (\mathcal G , L_2(Q) , \epsilon)$ is the smallest number of open balls of radius $\epsilon > 0$ required to cover $\mathcal G$. 
For the definition of VC classes, we follow \citet{gine+g:02}. Note that similar classes (sometimes with different names) have been defined earlier in the literature \citep{nolan+p:1987}. Next we call an envelop for $\mathcal G$ any function $G:S\mapsto \mathbb R$ that satisfies $|g(x)| \leq G(x)$ for all $x\in S$. 

\begin{definition}[VC-class]\label{def:VC}
	A class $\mathcal G$ of real functions on a measurable space $(S, \mathcal S)$ is called a VC-class of parameters $(v, A) \in (0, \infty)^2$ with respect to the envelope $G$ if for any $0 < \epsilon < 1$ and any probability measure $Q$ on $(S, \mathcal S)$, we have
	\begin{equation*}
	\mathcal N \left(\mathcal G,  L_2(Q)  ,  \epsilon \| G \| _ {L_2(Q) } \right) \le (A/\epsilon)^{v}.
	\end{equation*}
\end{definition}

The following concentration result is tailored to VC classes of functions. The result stated in Theorem \ref{vc_bound} below is a consequence of the work in \cite{gineConsistencyKernelDensity2001} which is based on \cite{talagrandNewConcentrationInequalities1996}. The next formulation is slightly different in that the role played by the VC constants ($v$ and $A$ below) is quantified. 
%Let $\mathcal F$ be a bounded class of measurable functions defined on $\mathcal X$. Let $U$ be a uniform bound for the class $\mathcal F$, i.e. $|f(x)| \leq U$ for all $f\in \mathcal F$ and $x\in \mathcal X$. %For notational simplicity and with no loss of generality, we require in the definition of a VC class that $A\geq 3\sqrt e$ and $v\geq 1$. 

\begin{theorem}\label{vc_bound}
Let $\mathcal G$ be a VC class of functions with parameters $(v,A)$ and uniform envelop $U>0$. Let $\sigma$ be such that $\sigma^2  \geq \sup_{g\in \mathcal G} \sigma^2_P (g) $ and $\sigma \geq U$.  Let $n\geq 1$ and $\delta\in (0,1) $ be such that
\begin{align*}
%\label{eq:constant_n_delta_1} 
 &\sqrt n \sigma \geq c_1  \sqrt {U^2 v\log(AU / (\sigma \delta)) },
\end{align*} 
then, we have with probability $1-\delta$,
\begin{align*}
&  \sup_{f\in \mathcal G} \left|  (P_n - P)  (g)\right|   \leq c_2 \sqrt { \frac{ \sigma^2 v \log( A U / (\sigma \delta) )}{n} }   ,
  \end{align*} 
where $c_1$ and $c_2$ are positive constants.
\end{theorem}

\paragraph{Preliminary results.}

Define $b>0$ and $U>0$ such that $b<f_X<U  $ and set
\begin{align*}
& \hat \tau_{n,k} (x) {=} \inf \{\tau\geq 0 \,:\, \sum_{i=1} ^ n  \mathds{1}_{ {B}(x,\tau)} (X_i)  \geq  k \},\\
&\overline{\tau}_{n,k}  = \left(\frac{ 2  k }{ n b_f V_d}  \right)^{1/ d},\\
 &\underline \tau_{n,k} = \left (\frac{   k  }{ 2n U_f V_d}  \right)^{1/ d},
\end{align*}
where $V_d$ is the volume of the unit ball in $\mathbb R^d$ and $0 < b_f\leq  U_f$ are constant factors that are defined right after. Note that $\hat B_k(x) = {B}(x,\hat \tau_{n,k} (x) ) $.  The next lemma is the starting point of our proof. As in \cite{jiang2019non}, it will be useful to control the size of the nearest neighbors balls. 
%from Lemma 2 in \cite{jiang2019non}.

\begin{lemma}\label{prop:balls}
Suppose that (B) is fulfilled and that $\log(n) / k \to 0$.  We have with probability $1$: there exists $N\geq 1$ such that for all $n\geq N$,
\begin{align*}
 \sup _{x\in [0,1]^d}  \sum_{i=1} ^n \mathds 1 _{ B  (x, \underline \tau_{n,k}  ) } (X_i ) \leq k\leq  \inf _{x\in [0,1]^d}  \sum_{i=1} ^n \mathds 1 _{ B  (x, \overline \tau_{n,k}  ) } (X_i ) .
\end{align*}

\end{lemma}

\begin{proof}
 Using (B) yields
\begin{align*}
P (X\in {B} (x, \overline \tau_{n,k} ) )  =   \int_{{B} (x, \overline \tau_{n,k}  ) \cap [0,1]^d } f_X  \geq b_f \lambda( {B} (x, \overline \tau_{n,k}  ) )  = b_fV_{d} \overline \tau_{n,k} ^{  d}  = 2k/n,
\end{align*}
for some constant factor $b_f>0$. Similarly, we have
\begin{align*}
 P (X\in {B} (x, \underline \tau_{n,k} ) )    \leq   k/(2n) 
\end{align*}
Applying Lemma \ref{lemma=prelim} with $\delta=n^{-2}$ and using that $ \log(n)/k \to 0$ and taking $n$ large enough,  we obtain that with probability at least $n^{-2}$, for all $x\in [0,1]^d$,
\begin{align*}
n^{-1} \sum_{i=1} ^n \mathds 1 _{ B  (x, \overline \tau_{n,k}  ) } (X_i ) & \geq  \frac{1}{2} P (B  (x, \overline \tau_{n,k} ) \geq k/n .
\end{align*}
As a consequence of the Borel-Cantelli lemma, we obtain that with probability $1$, for $n$ large enough,
\begin{align*}
\inf _{x\in [0,1]^d}  \sum_{i=1} ^n \mathds 1 _{ B  (x, \overline \tau_{n,k}  ) } (X_i ) \geq k .
\end{align*}
The upper bound is obtained similarly.
\end{proof}

\begin{lemma}\label{prop:tau}
Suppose that (B) is fulfilled and that $\log(n) / k \to 0$.  We have with probability $1$: there exists $N\geq 1$ such that for all $n\geq N$,
\begin{align*}
\underline \tau_{n,k} \leq  \inf _{x\in [0,1]^d} \hat \tau_{n,k} (x) \leq \sup _{x\in [0,1]^d} \hat \tau_{n,k} (x)   \leq \overline \tau_{n,k}
\end{align*}
\end{lemma}
\begin{proof}
By definition of $\hat \tau _{n,k}(x)$, on the event that $\sum_{i=1} ^n \mathds 1 _{ B  (x, \overline \tau_{n,k}  ) } (X_i )  \geq k$ it holds that $\hat \tau _{n,k} (x) \leq  \overline \tau _{n,k}$.
\end{proof}

\paragraph{End of the proof.}

We rely on the classical bias-variance decomposition
\begin{align*}
\hat F^{(NN)}  (y|x)  - F(y|x) = \sum_{i=1} ^n w_i ( \mathbb I_{\{Y_i\leq y \}}  - F(y|X_i) ) + \sum_{i=1} ^n w_i  ( F(y|X_i)    - F(y|x) ) ,
\end{align*}
where $w_i =k^{-1}  \mathds {1}_{\{\hat B_k(x)\}} (X_i) $. 
 We have
\begin{align*}
\left| \sum_{i=1} ^n w_i  ( F(Y_i|X_i)    - F(y|x) )  \right|  & \leq \sup_{x'\in \hat {B}_{n,k} (x ) } |  F (y|x' ) -  F (y|x)| \leq  L \hat \tau_{n,k} (x).
\end{align*}
Applying Lemma \ref{prop:tau} we obtain that, with probability $1$, for $n$ is large enough,
\begin{align*}
\left| \sum_{i=1} ^n w_i  ( F(Y_i|X_i)    - F(y|x) )  \right| \leq L \overline \tau_{n,k} .
\end{align*}
Besides, we write
\begin{align*}
\sum_{i=1} ^n w_i ( \mathds{1} _{\{Y_i\leq y \}}  - F(y|X_i) ) =   k^{-1} \sum_{i=1}^n  ( \mathds{1} _{\{Y_i\leq y \}}  - F(y|X_i) ) \mathds{1} _{{\hat B}_{k} (x)  }(  X_i) 
\end{align*}
Due to Lemma \ref{prop:tau}, we have for $n$ large enough,
\begin{align*}
& \sup_{ y\in \mathbb R, \, x\in [0,1] ^d }  \left|  \sum_{i=1}^n  ( \mathds{1} _{\{Y_i\leq y \}}  - F(y|X_i) ) \mathds{1} _{{\hat B}_{k} (x)  }(  X_i)\right |  \\
& \leq\sup_{ y\in \mathbb R, \, B\in  B(\overline \tau_{n,k} )  }  \left|  \sum_{i=1}^n  ( \mathds{1} _{\{Y_i\leq y \}}  - F(y|X_i) ) \mathds{1} _{ B  }(  X_i)\right | ,
\end{align*}
where $ \mathcal B(\overline \tau_{n,k} ) $ is the set of all balls having radius smaller than  $\overline \tau_{n,k}$. Hence, the class of interest is 
\begin{align*}
\mathcal G = \{ (y,x) \mapsto (\mathds 1_{\{Y_i\leq y \}}  - F(y|X_i) )  \mathbb I _{  B    }(  X_i) \, :\, y\in \mathbb R,\, B \in \mathcal B(\overline \tau_{n,k} )  \}. 
\end{align*}
The class of cells $\mathds 1_{\{Y_i\leq y \}}$, $y\in \mathbb R$ is VC with parameter $v$ independent of the dimension \citep{wellner1996}. The class of balls $\mathcal B$ is VC with $v\propto d$ \citep{wenocur1981some}. It follows that their product is of  VC-type with dimension proportional to $d$.  Noting that $E [ \mathds 1_{\{Y_1\leq y \}}  \mathbb I _{  B    }(  X_1) |X_1] =  F(y|X_1)  \mathbb I _{  B    }(  X_1) $ we can use Lemma 20 in \cite{nolan+p:1987} to prove that $ F(y|X_1)  \mathbb I _{  B    }(  X_1) $ is VC with $v\propto d$. Finally, an appeal to Lemma 16  \cite{nolan+p:1987}  implies that the class $\mathcal G$ is VC with constant $v\propto d$ and envelop $U=2$.   Consequently,  we  can apply Theorem \ref{vc_bound} with $U = 2$. Moreover
\begin{align*}
 \sup_{y\in \mathbb R,\, B \in \mathcal B(\overline \tau_{n,k} ) }  E [   ( \mathbb I_{\{Y_1\leq y \}}  - F(y|X_1) ) ^2 \mathbb I _{  B  (x, \tau)  }(  X_1)^2 ]
&\leq  \sup_{ B \in \mathcal B(\overline \tau_{n,k} )}  E [   \mathbb I _{  B     }(  X_1)].
\end{align*}
In a similar way as in the proof of Lemma \ref{prop:balls}, we find
\begin{align*}
P (X\in {B} (x, \overline \tau_{n,k} ) )    \leq  U_fV_{d} \overline \tau_{n,k} ^{  d}  = ( U_f / b_f) (2 k/n).
\end{align*}
Hence  we can take  $\sigma ^2 =   ( U_f / b_f) (2 k/n)$ in our application of Theorem \ref{vc_bound}.
Because $\log (n) / k \to 0$, the condition on $(n,\sigma,\delta, U)$ is easily satisfied with $\delta = n^{-2}$ and $n$ sufficiently large. We get that with probability $1/n^2$, 
\begin{align*}
&\sup_{ y\in \mathbb R, \, B\in \mathcal B(\overline \tau_{n,k} )  }  \left|  \sum_{i=1}^n  ( \mathds{1} _{\{Y_i\leq y \}}  - F(y|X_i) ) \mathds{1} _{ B  }(  X_i)\right |  \leq  C \sqrt {  k d \log(n) } . 
\end{align*}
for some constant $C$ that depends on $(U_f,b_f)$. The conclusion comes invoking the Borel-Cantelli Lemma.

\end{appendices}

\end{document}